\documentclass[11pt,a4paper]{article}

\usepackage{amsmath,amssymb,amsfonts,amsthm}
\usepackage{mathtools}
\usepackage{bbm}
\usepackage{graphicx}
\usepackage{subcaption}
\usepackage{verbatim}

\newcommand{\safeincludegraphics}[2][]{%
\IfFileExists{#2}{%
\includegraphics[#1]{#2}%
}{%
\fbox{\parbox[c][0.28\textheight][c]{0.85\linewidth}{\centering Missing figure:\\ \texttt{\detokenize{#2}}}}%
}%
}

\usepackage{hyperref}
\usepackage{enumitem}
\usepackage{algorithm}
\usepackage{tikz}
\usepackage{geometry}
\usetikzlibrary{arrows.meta, positioning}
\newtheorem{theorem}{Theorem}[section]
\newtheorem{proposition}[theorem]{Proposition}
\newtheorem{lemma}[theorem]{Lemma}

\theoremstyle{definition}

\newtheorem{definition}[theorem]{Definition}

\theoremstyle{remark}
\newtheorem{remark}[theorem]{Remark}

\newcommand{\E}{\mathbb{E}}
\newcommand{\cE}{\mathcal{E}}

\newcommand{\C}{\mathcal{C}}
\newcommand{\F}{\mathcal{F}}

\newcommand{\R}{\mathbb{R}}

\DeclareMathOperator*{\essinf}{ess\,inf}
\DeclareMathOperator*{\esssup}{ess\,sup}

\begin{document}

\title{Risk Measures under Paired-Ambiguity: A Deep Learning Reflected BSDE Framework}
\author{Nacira Agram $^{1}$, Jan Rems $^{2}$ and Emanuela Rosazza Gianin $^{3}$}
\date{\today}
\maketitle

\footnotetext[1]{Department of Mathematics, KTH Royal Institute of Technology and Digital Futures, 100 44, Stockholm, Sweden. 
Email: nacira@kth.se. Work supported by the Swedish Research Council grant (2020-04697).}

\footnotetext[2]{Department of Mathematics, University of Ljubljana, Ljubljana, Slovenia.  Email: jan.rems@fmf.uni-lj.si. Work supported by Slovenian Research and Innovation Agency, research core funding No.P1-0448.}

\footnotetext[3]{Department of Statistics and Quantitative Methods, University of Milano-Bicocca, Milano, Italy. Email:  emanuela.rosazza1@unimib.it. Member of GNAMPA, INdAM, Italy. Work supported by Gnampa Research Project 2024 (PRR-20231026-073916-203).}

\begin{abstract}
We study optimal stopping under dynamic risk measures with simultaneous
ambiguity in the probability model and the discount rate. We introduce a
paired ambiguity framework combining Girsanov model uncertainty with cash
subadditive risk evaluation and characterize the stopping value by an upper
reflected backward stochastic differential equation (BSDE). We establish
structural properties of the resulting stopping operator and study quadratic
drivers associated with entropic risk measures, obtaining explicit stopping
rules in several benchmark cases. We then develop a deep learning scheme for
the reflected quadratic BSDE. The convergence analysis uses discrete
reflection and truncation to reduce the quadratic problem to a globally
Lipschitz system and combines reflected BSDE discretization estimates with
neural network approximation errors. Numerical experiments for American
options illustrate the effects of discount rate and entropic ambiguity on
stopping values and exercise decisions.

\end{abstract}

\section{Introduction}

Dynamic risk measures have become a fundamental tool for the quantitative
assessment of financial positions under uncertainty (see, among many others, \cite{ADEHK, delb, fritt-rg, CDK, barrieu-elk, bion-nadal, detl-scand}). Unlike classical conditional
expectations, dynamic risk measures allow nonlinear preferences, model
uncertainty, recursive evaluation and time-consistent risk assessment. Following
the pioneering works on BSDEs and $g$-expectations
\cite{ELKPQ,peng,CHMP,BCHMP}, it is now well understood that a large class of
dynamic risk measures admits a representation through backward stochastic
differential equations (BSDEs), where the generator completely characterizes the
underlying notion of risk, ambiguity and investor preferences. This connection
has subsequently been developed in the context of nonlinear pricing, coherent and convex risk measures, and utility valuation
\cite{ELK-quenez,rg,barrieu-elk, DPR}.

Among the various classes of dynamic risk measures, cash-subadditive risk
measures have attracted considerable attention because they naturally incorporate
ambiguity in discounting and interest rates. Unlike the classical
cash-additive framework, they allow uncertainty in the time value of money and
therefore provide a more flexible approach to long term financial valuation.
Their BSDE characterization and dual representation were established by El
Karoui and Ravanelli \cite{ELK-rav}, who showed that discount rate ambiguity
induces a nonlinear dependence of the driver on the value process. See also, e.g., \cite{dinunno-rg, mastrog-rg} for a discussion on cash-subadditive risk measures.

Optimal stopping under dynamic risk measures extends the classical theory of
optimal stopping by replacing conditional expectations with nonlinear dynamic
risk evaluations. Such problems arise naturally in the valuation of
American derivatives under model uncertainty, optimal liquidation,
insurance contracts, real options, and investment timing problems. In the
Brownian framework, Bayraktar, Karatzas and Yao
\cite{bayrakt-etal} studied optimal stopping for convex dynamic risk measures,
while Quenez and Sulem \cite{quenez-sulem} established a reflected BSDE
representation for robust optimal stopping under dynamic risk measures with
jumps, thereby extending the nonlinear Snell envelope approach initiated by El
Karoui et al. \cite{ELK-etal_Reflect}.

The objective of the present paper is to develop a unified framework for optimal
stopping under \emph{paired ambiguity}, where uncertainty simultaneously affects
the probabilistic model and the discount mechanism. The first source of
ambiguity is represented through changes of probability measure via Girsanov
transformations, while the second corresponds to uncertainty in stochastic
discounting, following the cash-subadditive framework of
\cite{ELK-rav}. Combining these two mechanisms leads to a new family of dynamic
risk measures (called \textit{paired-ambiguity dynamic risk measures}) that naturally generalizes both classical entropic risk measures
and cash-subadditive risk measures.

Our first contribution is to formulate the corresponding optimal stopping
problem and to characterize its value process by means of an upper reflected
BSDE. The upper obstacle naturally appears because the problem is formulated in
terms of liabilities, so that immediate stopping corresponds to the negative of
the payoff process. We establish a reflected BSDE representation theorem and
show that the optimal stopping strategy is given by the first hitting time of
the reflected solution against the obstacle.

Our second contribution concerns the structural properties of the resulting
stopping operator. We prove that it inherits the principal axioms of dynamic
risk measures and investigate its geometric properties, including monotonicity,
cash-subadditivity, cash-additivity, quasi-concavity, concavity under affine
drivers, and the failure of convexity in general. These results indicate that
the stopping operator should be viewed as a dynamic risk measure acting on
adapted stochastic processes rather than merely on terminal random variables.

Special attention is devoted to quadratic drivers associated with entropic risk
measures. Quadratic BSDEs play a central role in mathematical finance through
their connection with exponential utility, entropic risk measures and robust
valuation \cite{kobyl-bsde,barrieu-elk, dinunno-rg}. We revisit this framework in the reflected
setting and derive the corresponding representation of entropic stopping risk
measures. In particular, we obtain explicit stopping rules in several tractable
situations and illustrate how discount ambiguity and model ambiguity interact
within the reflected BSDE framework. The existence and uniqueness of reflected
quadratic BSDEs rely on the theory developed by Kobylanski et al.
\cite{kobyl-etal}.

Finally, we propose a deep learning approximation scheme for solving the
resulting reflected BSDEs. Since the drivers considered in this paper may
exhibit quadratic growth in the control variable, standard convergence results
for Lipschitz BSDEs are no longer directly applicable. Deep schemes for reflected BSDEs with
Lipschitz drivers have been developed both by hard projection \cite{hure-etal}
and by regularization \cite{hu-zou}; for quadratic drivers without reflection
the truncation analysis of \cite{chassagneux-richou} is available, and the
reflected quadratic case has been treated in \cite{sun-liang-tang} for a scheme
in which the conditional expectations are computed exactly. We combine the last
two, the truncation being made exact by the uniform bound on the integrand of
the discretely reflected equation, and obtain a convergence estimate in which
the network approximation errors appear explicitly.
Numerical experiments for American options demonstrate the practical
performance of the method and illustrate the impact of ambiguity on optimal
stopping decisions.

The paper is organized as follows. Section~2 introduces the paired-ambiguity
dynamic risk measure and its dual representation. Section~3 formulates the
optimal stopping problem and derives the reflected BSDE characterization
together with the main properties of the stopping operator. Section~4 is devoted
to the entropic case and explicit stopping representations. Section 5 develops the deep learning approximation scheme and establishes
its convergence. Section 6 presents financially meaningful examples of
dynamic risk measures generated by reflected BSDEs. Finally, Section 7
reports numerical experiments illustrating the theoretical results.

\section{Paired-Ambiguity Risk Measures and Reflected BSDE Representation}
\label{sec:stopping}

In this section we formulate the optimal stopping problem associated with the
paired-ambiguity dynamic risk measure introduced here below and
derive its reflected BSDE representation. The framework combines two different
sources of uncertainty. The first corresponds to ambiguity in the underlying
probabilistic model and is represented through Girsanov changes of measure. The
second corresponds to ambiguity in stochastic discounting, leading to a
cash-subadditive dynamic evaluation. As a consequence, the resulting stopping
problem extends the classical nonlinear optimal stopping framework to a setting
where both model uncertainty and discount rate uncertainty are treated
simultaneously.

The formulation is inspired by the optimal stopping problem for convex dynamic
risk measures studied by Bayraktar, Karatzas and Yao
\cite{bayrakt-etal}, while incorporating the cash-subadditive dual
representation introduced by El Karoui and Ravanelli \cite{ELK-rav}. The
combination of these two ingredients leads naturally to an upper reflected BSDE,
which characterizes both the stopping value and the optimal exercise strategy.

\subsection{Paired-ambiguity optimal stopping problem}

Consider a bounded adapted payoff process
\[
\xi=(\xi_t)_{0\le t\le T}.
\]
For every stopping time \(\tau\in\mathcal T_t\), the paired-ambiguity dynamic
risk measure is defined by
\begin{equation}\label{eq:rho-csa-g}
\rho_{t,\tau}(\xi_\tau)=
\sup_{(\alpha,\beta)\in\mathcal A}
\mathbb E^{\mathbb Q^\alpha}
\left[
e^{-\int_t^\tau\beta_s\,ds}
(-\xi_\tau)
-
\int_t^\tau
e^{-\int_t^s\beta_u\,du}
g^*(s,\alpha_s,\beta_s)\,ds
\;\middle|\;
\mathcal F_t
\right],
\end{equation}
where the process \(\alpha\) models ambiguity in the probability measure through
Girsanov transformations, \(\beta\) represents ambiguity in stochastic
discounting, and $\mathcal{A}$ is a given set over which there is paired-ambiguity. The function \(g^*\) denotes the convex conjugate of the BSDE driver \(g\).

The corresponding optimal stopping problem is
\begin{equation}\label{eq:optimal stopping}
V_t
=
\essinf_{\tau\in\mathcal T_t}
\rho_{t,\tau}(\xi_\tau).
\end{equation}

Unlike the classical cash-additive framework, the evaluation
\eqref{eq:rho-csa-g} simultaneously accounts for uncertainty in the probabilistic
model and in the discount mechanism. Consequently, the stopping value
incorporates both model ambiguity and interest rate ambiguity within a unified
dynamic risk measure framework. Note that \eqref{eq:rho-csa-g} corresponds to the dual representation of El Karoui and Ravanelli \cite{ELK-rav} for $\mathcal{A}$ as in \eqref{eq: A}. Here, instead, $\mathcal{A}$ is considered general.

The representation of the value process follows the same general philosophy as
nonlinear Snell envelopes generated by BSDEs. However, the presence of the
discount process modifies the dual representation and changes the associated BSDE
driver. The reflected BSDE obtained below therefore extends the standard
representation by incorporating both sources of ambiguity into a single
nonlinear evaluation.

\subsection{Standing assumptions} \label{sec: assumpt}

We work on a complete filtered probability space
$(\Omega,\mathcal F,(\mathcal F_t)_{0\le t\le T},\mathbb P)$ satisfying the
usual conditions and carrying a $d$-dimensional Brownian motion $W$. For
$t\in[0,T]$, let $\mathcal T_t$ denote the set of
$(\mathcal F_s)$-stopping times with values in $[t,T]$.
We use the standard spaces
\[
\mathcal S^2
=
\left\{Y:\ Y\text{ is adapted and continuous, }
\mathbb E\left[\sup_{0\le t\le T}|Y_t|^2\right]<\infty\right\}
\]
and
\[
\mathcal H^2
=
\left\{Z:\ Z\text{ is progressively measurable and }
\mathbb E\left[\int_0^T |Z_t|^2\,dt\right]<\infty\right\}.
\]

Let $\xi=(\xi_t)_{t\in[0,T]}$ be an adapted bounded continuous process with
$\xi\in\mathcal S^2$. The process $\xi$ represents the financial payoff received
upon stopping, so the upper obstacle for the risk process is $-\xi$ and the
terminal condition is compatible with the obstacle, namely $Y_T=-\xi_T$. Let
\[
g:\Omega\times[0,T]\times\mathbb R\times\mathbb R^d\to\mathbb R
\]
be a progressively measurable driver.
We assume the following:
\begin{itemize}
\item $g(t, \cdot, \cdot)$ is continuous with linear-quadratic growth, i.e. there exists a constant $C>0$ such that for any $y \in \mathbb{R}$, $z \in \R^d$ 
\begin{equation*}
|g(t,y,z)|  \leq |g(t,0,0)| +C |y| + C |z|^2, \qquad
g(\cdot,0,0)\in\mathcal H^2.
\end{equation*}
\item For any $M >0$ there exists $C>0$ such that for any $y \in [-M,M]$, $z \in \R^d$ 
\begin{equation*}
|\frac{\partial g}{\partial z}(t,y,z)|  \leq C (1+ |z|).
\end{equation*}
\item For any $\epsilon >0$ there exists $C_{\epsilon}>0$ such that for any $y \in \R$, $z \in \R^d$ 
\begin{equation*}
\frac{\partial g}{\partial y}(t,y,z) \leq C_{\epsilon} + \epsilon |z|^2.
\end{equation*}
\end{itemize}
Note that the last assumption is automatically satisfied if $g$ is differentiable and decreasing in $y$.

The assumptions guarantee that the fixed horizon BSDEs and the upper
reflected BSDE below have a unique solution and the Comparison Theorem holds for both the BSDEs and the reflected BSDEs. See Theorems 2.3 and 2.6 of Kobylanski \cite{kobyl-bsde} for BSDEs and Corollary 1 and Proposition 3.2 of Kobylanski et al. \cite{kobyl-etal} for reflected BSDEs.
Consider a BSDE with driver $g$ of the form
\begin{equation}\label{eq: risk measure-bsde}
Y_t
=
-\xi_T
+
\int_t^T g(s,Y_s,Z_s)\,ds
-
\int_t^T Z_s\,dW_s.
\qquad 0\le t\le T,
\end{equation}

We recall the following definition of (upper) reflected BSDEs. See \cite{ELK-etal_Reflect, kobyl-etal} and the references therein.

\begin{definition}[Upper reflected BSDE]
\label{def:upper-rbsde}
Let \(L=(L_t)_{0\le t\le T}\) be an adapted obstacle process and let
\(\eta\) be an \(\mathcal F_T\)-measurable terminal condition satisfying
\[
\eta\le L_T.
\]
A triple \((Y,Z,K)\) is called a solution of the upper reflected BSDE
with generator \(g\), terminal condition \(\eta\), and upper obstacle \(L\)
if
\begin{equation}
\label{eq:upper-rbsde}
Y_t
=
\eta
+
\int_t^T g(s,Y_s,Z_s)\,ds
-
\int_t^T Z_s\,dW_s
-
(K_T-K_t),
\qquad 0\le t\le T,
\end{equation}
where \(K\) is an adapted nondecreasing process with \(K_0=0\), such that
\begin{equation}
\label{eq:upper-obstacle}
Y_t\le L_t,
\qquad 0\le t\le T,
\end{equation}
and
\begin{equation}
\label{eq:upper-skorokhod}
\int_0^T (Y_t-L_t)\,dK_t=0.
\end{equation}
The last condition is the Skorokhod minimality condition; equivalently,
\(K\) can increase only on the contact set
\[
\{t\in[0,T]:Y_t=L_t\}.
\]
\end{definition}

\section{Optimal stopping problem for risk measures and obstacle}

Let $g^*$ denote the convex conjugate associated with the dual variables
corresponding to model and discount ambiguity. Following the dual representation
of cash-subadditive risk measures \cite{ELK-rav}, we consider admissible pairs
$(\alpha,\beta)$ in the class
\begin{equation} \label{eq: A}
\mathcal A
=
\left\{
(\alpha,\beta):
\alpha\in \mathrm{BMO},\quad
0\le \beta_t\le C,\quad
g^*(t,\alpha_t,\beta_t)<+\infty
\right\}.
\end{equation}
For each admissible $\alpha$, the stochastic exponential
\[
\mathcal E\left(-\int_0^\cdot \alpha_s\,dW_s\right)_T
=
\exp\left\{
-\int_0^T \alpha_s\,dW_s
-\frac12\int_0^T |\alpha_s|^2\,ds
\right\}
\]
is assumed to be a uniformly integrable martingale and defines an
equivalent probability measure $\mathbb Q^\alpha$ by
\[
\frac{d\mathbb Q^\alpha}{d\mathbb P}
=
\mathcal E\left(-\int_0^\cdot \alpha_s\,dW_s\right)_T.
\]
The process $\beta$ represents ambiguity in the stochastic discount rate. We also
assume that the discounted penalty term is integrable for every
$(\alpha,\beta)\in\mathcal A$, and that $\mathcal A$ is stable under stopping
and pasting: if two admissible pairs are used before and after a stopping time,
then the pasted pair is again admissible. Finally, the dual representation
below is assumed to be time consistent, equivalently generated by the BSDE
driver in \eqref{eq: risk measure-bsde}. These stability and time consistency
conditions are the hypotheses needed for the nonlinear Snell envelope and
reflected BSDE arguments.

For $t\le \tau$, the paired-ambiguity risk evaluation is
\begin{equation}\label{eq:section3-risk-eval}
\rho_{t,\tau}(\xi_\tau)
=
\sup_{(\alpha,\beta)\in\mathcal A}
\mathbb E^{\mathbb Q^\alpha}
\left[
 e^{-\int_t^\tau \beta_s\,ds}(-\xi_\tau)
-
\int_t^\tau e^{-\int_t^s\beta_u\,du}
g^*(s,\alpha_s,\beta_s)\,ds
\;\middle|\;\mathcal F_t
\right],
\end{equation}
where 
\begin{equation} \label{eq: g-star}
g^*(s,\alpha,\beta)= \sup_{(y,z) \in \mathbb R \times \mathbb R^d} \{-\beta y - \alpha \cdot z - g(s,y,z) \}.
\end{equation}
See Theorem 7.5 of
\cite{ELK-rav}.


Under the standing assumptions, we now focus on the following optimal stopping problem
\begin{equation} \label{eq: inf-rho}
\varphi_t(\xi):=  \essinf_{\tau\in\mathcal T_t}\rho_{t,\tau}(\xi_\tau).
\end{equation}

The following result is the upper obstacle, risk measure counterpart of
Proposition 3.1 in \cite{kobyl-etal}. After the corresponding sign
transformation, the essential supremum characterization therein becomes
the essential infimum characterization below.

\begin{theorem}[Reflected BSDE representation]
\label{thm:rbsde-representation}
Under the standing assumptions \ref{sec: assumpt}, the value process $(\varphi_t(\xi))_t$ in \eqref{eq: inf-rho}
is the first component of the unique upper reflected BSDE
\eqref{eq:upper-rbsde}--\eqref{eq:upper-skorokhod}, with
\[
\eta=-\xi_T,
\qquad
L_t=-\xi_t.
\]
In particular,
\[
\varphi_t(\xi)=Y_t.
\]

Moreover, under the regularity assumptions ensuring attainment of the
optimal stopping problem, the first contact time
\[
\tau_t^*
=
\inf\{s\ge t:Y_s=-\xi_s\}
\]
is optimal.
\end{theorem}

Note that the stopping payoff is received as a financial payoff $\xi_t$, while the risk
process is written with the sign convention of a liability. Hence the immediate
stopping risk is $-\xi_t$, and the reflected BSDE has the upper obstacle $(-\xi_t)$.
\medskip

In the following result, we prove the properties satisfied by the optimal stopping value $\varphi$.

For \(m_t\in L^\infty(\mathcal F_t)\), we use the notation
\(\xi+m_t\) for the shifted payoff process defined on \([t,T]\) by
\[
(\xi+m_t)_s:=\xi_s+m_t,
\qquad s\in[t,T].
\]
\smallskip

\begin{theorem}[Dynamic risk measure properties and geometry of the stopping operator]\label{prop: properties} The stopping operator $\varphi$ defined in \eqref{eq: inf-rho} inherits the main axioms of dynamic risk measure. More precisely, it satisfies the following properties.

\noindent a) (decreasing monotonicity) For any $t \in [0,T]$, $\varphi_t$ is monotone decreasing, i.e., if $\xi_s \geq \eta_s$ for any $s \in [0,T]$ (with $\eta$ being a bounded adapted obstacle), then $\varphi_t(\xi) \leq \varphi_t(\eta)$.

\noindent b) (cash-subadditivity) If $g$ is decreasing in $y$, then $\varphi_t$ is cash-subadditive for any $t \in [0,T]$, i.e., $\varphi_t(\xi + m_t) \geq \varphi_t(\xi)-m_t$ for any obstacle $\xi$ and $m_t \in L^\infty_+(\F_t)$.

\noindent c) (cash-additivity) If $g$ is independent of $y$, then $\varphi_t$ is cash-additive for any $t \in [0,T]$, i.e., $\varphi_t(\xi + m_t) = \varphi_t(\xi)-m_t$ for any obstacle $\xi$ and $m_t \in L^\infty(\F_t)$.

\noindent d) (quasi-concavity) Assume that any $\rho_{t, \tau}$ satisfies the following continuity property: $$\rho_{t,\tau}(\esssup\{X;Y\}) = \essinf\{ \rho_{t,\tau}(X); \rho_{t,\tau}(Y) \}.$$ 

If $\varphi_t(\xi) \geq c_t$ and $\varphi_t(\eta) \geq c_t$ for bounded obstacles $\xi, \eta$ and for some $t \in [0,T]$ and $c_t \in L^{\infty}(\F_t)$, then $\varphi_t(\lambda \xi + (1-\lambda) \eta) \geq c_t$ for any $\lambda \in [0,1]$.

\noindent e) If $g$ is affine in $y,z$, i.e. $g(t,y,z)=ay+b \cdot z+c$ for some $a,c \in \mathbb R$ and $b \in \mathbb R^d$, then $\varphi_t$ is concave for any $t \in [0,T]$, i.e., $\varphi_t(\lambda \xi + (1-\lambda) \eta) \geq \lambda \varphi_t(\xi)+ (1-\lambda) \varphi_t(\eta)$ for all obstacles $\xi, \eta$ and $\lambda \in [0,1]$.

\noindent f) $\varphi_t$ is not convex in general. For instance, convexity fails when $g(t,y,0)=0$ for any $y \in \mathbb{R}$ or when $g$ is independent of $y$ and $g(t,0)=0$.
\end{theorem}

The previous result highlights that convexity of $\varphi$ fails in general and that the sufficient conditions (e.g., the continuity assumption in d) and the assumption on $g$ in e)) for $\varphi$ to be (quasi-)concave are quite strong.

\begin{proof}

\noindent a) Let $\xi, \eta$ be adapted and bounded obstacles with $\xi_s \geq \eta_s$ for any $s \in [0,T]$.

By the properties of BSDEs and of the associated risk measures (see, e.g., \cite{barrieu-elk}, \cite{ELKPQ}, \cite{rg}), $\rho^g_{t,\tau}(\xi_\tau) \leq \rho^g_{t, \tau} (\eta_\tau)$ for any $t \leq \tau$.
It then follows that $\varphi_t(\xi) \leq \varphi_t(\eta)$ for all $t \in [0,T]$.

\noindent b) By El Karoui and Ravanelli \cite{ELK-rav}, decreasing monotonicity of $g$ in $y$ implies cash-subadditivity of any $\rho_{t,\tau}$. It then follows that for any obstacle $\xi$ and any $m_t \in L^\infty_+(\F_t)$
\begin{align*}
\varphi_t(\xi + m_t)&= \essinf_{\tau\in\mathcal T_t}\rho_{t,\tau}(\xi + m_t) \\
&\geq \essinf_{\tau\in\mathcal T_t} \{ \rho_{t,\tau}(\xi ) - m_t \} \\
&= \varphi_t(\xi) - m_t.
\end{align*}

\noindent c) By Briand et al. \cite{BCHMP}, independence of $g$ in $y$ implies cash-additivity of any $\rho_{t,\tau}$. The thesis then follows by proceeding as in item b).

\noindent d) Clearly, $\varphi_t(\xi) \geq c_t$ implies that $\rho_{t,\tau}(\xi_\tau)\geq c_t$ for any $\tau \in \mathcal T_t$; similarly for $\varphi_t(\eta)$.

For any $\lambda \in [0,1]$
\begin{align}
\varphi_t(\lambda \xi + (1-\lambda) \eta)&= \essinf_{\tau\in\mathcal T_t}\rho_{t,\tau}(\lambda \xi_\tau + (1-\lambda) \eta_\tau)  \notag \\
&\geq \essinf_{\tau\in\mathcal T_t} \rho_{t,\tau}(\esssup\{\xi_\tau; \eta_\tau \})  \notag \\
&= \essinf_{\tau\in\mathcal T_t} \essinf\{ \rho_{t,\tau}(\xi_\tau); \rho_{t,\tau}(\eta_\tau) \} \label{eq: cont} \\
&\geq c_t, \notag
\end{align}
where \eqref{eq: cont} is due to the continuity assumption of $\rho_{t,\tau}$.

\noindent e) If $g$ is affine, then the associated $\rho_{t,s}$ is both convex and concave (see, e.g., \cite{barrieu-elk}, \cite{ELKPQ}, \cite{rg}).
It then follows that, for any $\lambda \in [0,1]$ and all obstacles $\xi, \eta$,
\begin{align*}
\varphi_t(\lambda \xi + (1-\lambda) \eta)&= \essinf_{\tau\in\mathcal T_t}\rho_{t,\tau}(\lambda \xi_\tau + (1-\lambda) \eta_\tau)  \\
&\geq \essinf_{\tau\in\mathcal T_t} \{ \lambda \rho_{t,\tau}(\xi_\tau) + (1-\lambda) \rho_{t, \tau} (\eta_\tau) \} \\
&\geq \lambda \essinf_{\tau\in\mathcal T_t}   \rho_{t,\tau}(\xi_\tau) + (1-\lambda) \essinf_{\tau\in\mathcal T_t}\rho_{t, \tau} (\eta_\tau) \\
&= \lambda \varphi_t( \xi) + (1-\lambda) \varphi_t(\eta).
\end{align*}

\noindent f) Assume that $g(t,y,0)=0$ for any $y \in \mathbb{R}$ or that $g$ is independent of $y$ and $g(t,0)=0$. Then, by \cite{peng} (see also Prop. 2.1 of \cite{CHMP}), the associated $\rho_{t,s}$ satisfies constancy, i.e. $\rho_{t,s}(c_t)=-c_t$ for any $c_t \in L^{\infty}(\F_t)$.

For deterministic obstacles $\xi, \eta$, then,
\begin{equation*}
\varphi_t(\xi)=\essinf_{\tau\in\mathcal T_t}  \rho_{t,\tau}(\xi_\tau)= \essinf_{s \in [t,T]} \rho_{t,s}(\xi_s)= \essinf_{s \in [t,T]} \, (-\xi_s),
\end{equation*}
where the last equality is due to the constancy property. A similar expression holds for $\varphi_t(\eta)$.

Let $t \in [0,T)$ and take $\xi_s=s$, $\eta_s=-2s$ for any $s \in [t,T]$ and $\lambda= \frac 12$. It then follows immediately that
\begin{equation*}
\varphi_t(\xi)= \essinf_{s \in [t,T]} \, (-\xi_s)=-T; \quad \varphi_t(\eta)= \essinf_{s \in [t,T]} \, (-\eta_s)=2t
\end{equation*}
and
\begin{equation*}
\varphi_t(\lambda \xi + (1-\lambda) \eta)  =\essinf_{s \in [t,T]} \, \left(- \frac 12 \xi_s - \frac 12 \eta_s \right) =\frac{t}{2} > \frac 12 \varphi_t( \xi) + \frac 12 \varphi_t(\eta)= -\frac{T}{2}+t,
\end{equation*}
where the strict inequality holds for any $t <T$. Convexity then fails in general.
\end{proof}

Note that the previous proof is based on the relation between optimal stopping problems, reflected BSDEs, and BSDEs. 

It is worth underlining that, by its definition and by its properties, $\varphi$ can be seen more as a dynamic risk measure for processes than as one for random variables.

In the following, we will focus our attention on examples of drivers of the following form:
\begin{equation*}
g(t,y,z)= \sup_{(\delta, \gamma) \in \C}  \psi (\delta, \gamma, y,z),  
\end{equation*}
for a given set $\C$ that can be interpreted as the set over which there is paired-ambiguity and for $\psi(\delta,\gamma, y,z)$ convex in $(y,z)$ and decreasing in $y$. Consequently, the associated driver $g$ satisfies the same properties.
\bigskip

Some special cases of drivers could be:
\begin{itemize}
    \item $g(t,y,z)= \sup_{(\delta, \gamma) \in \C}  \{-\delta y + \frac{\gamma}{2} |z|^2 \}$ with $\delta >0$ on $\C$ that can be seen as a paired-ambiguity on aversion in the entropic-type term and on discounting;
\item $g(t,y,z)= \sup_{\delta \in \C}  \{-\delta y \}$ with $\delta >0$ on $\C$ (ambiguity on discounting);
\item $g(t,y,z)= \sup_{\gamma \in \C}  \{\gamma z \}$ (ambiguity on the probability measure/generalized scenario $Q$ over which expectation is taken);
\item $g(t,y,z)= \sup_{(\delta, \gamma) \in \C} \{ -\delta y + \gamma z \}$ with $\delta >0$ on $\C$ (linear pair-ambiguity);
\item $g(t,y,z)= \sup_{\delta}  \{-\delta y + \bar{\gamma}z\}$ with $\delta >0$ on $\C$;
\item $g(t,y,z)= \sup_{\gamma}  \{ \frac{\gamma}{2} |z|^2 \}$ (ambiguity on the aversion parameter in the entropic case).
\end{itemize}

\section{Entropic Risk Measure}\label{sec:entropic}

The entropic case is the canonical reflected quadratic BSDE example for this
stopping problem. It should be treated separately from the Lipschitz theory: existence, uniqueness and comparison rely
on quadratic BSDE estimates, BMO martingales and exponential integrability
conditions as in Kobylanski type results \cite{kobyl-bsde,kobyl-etal}. In the
reflected case, these conditions are combined with the upper obstacle estimates
used in the Quenez--Sulem representation \cite{quenez-sulem}.

\subsection{Properties of the entropic reflected quadratic BSDE}

In the stopping problem, the relevant object is not the fixed terminal entropic
BSDE alone, but the upper reflected quadratic BSDE with obstacle $-\xi$:
\begin{equation}\label{eq:entropic-rbsde-properties}
Y_t
=
-\xi_T
+
\int_t^T \frac{\gamma}{2}|Z_s|^2\,ds
-
\int_t^T Z_s\,dW_s
-
(K_T-K_t),
\qquad 0\le t\le T,
\end{equation}
with
\begin{equation}\label{eq:entropic-rbsde-properties-obstacle}
Y_t\le -\xi_t,
\qquad
\int_0^T (Y_t+\xi_t)\,dK_t=0.
\end{equation}
Here $K$ is the reflection process. The ordinary entropic BSDE appears only as
the local equation on an interval where no reflection occurs, or as the
fixed stopping time evaluation used inside the nonlinear Snell envelope.

\subsection{Explicit entropic representation and reflected envelope}
We now exploit the exponential structure of the entropic driver to obtain
a more explicit representation of the stopping problem. The exponential
transformation converts the upper reflected quadratic BSDE into a classical
lower Snell envelope problem for the transformed payoff $e^{-\gamma\xi}$.
This yields both an explicit representation of the value process and a
direct characterization of the associated optimal stopping rule.
\begin{proposition}[Exponential transform for the entropic reflected quadratic BSDE]
Let $\gamma>0$, and assume that the payoff process $\xi$ is continuous and has
sufficient exponential moments so that the quadratic reflected BSDE below is
well posed in the sense of the quadratic reflected BSDE literature. Let
$g(z)=(\gamma/2)|z|^2$, and denote by
$\cE_{t,\tau}^g[\eta]$ the first component at time $t$ of the BSDE on
$[t,\tau]$ with driver $g$ and terminal condition $\eta$. In the entropic case,
\begin{equation}\label{eq:entropic-risk-eval}
\cE_{t,\tau}^g[-\xi_\tau]
=
\frac{1}{\gamma}
\log\mathbb{E}\left[e^{-\gamma\xi_\tau}\mid\mathcal{F}_t\right].
\end{equation}
Then the value process
\begin{equation}\label{eq:entropic-stopping-value}
V_t
:=
\essinf_{\tau\in\mathcal T_t}
\cE_{t,\tau}^g[-\xi_\tau]
\end{equation}
is characterized by the upper reflected quadratic BSDE
\begin{equation}\label{eq:entropic-reflected-bsde}
V_t
=
-\xi_T
+
\int_t^T
\frac{\gamma}{2}|Z_s|^2\,ds
-
\int_t^T Z_s\,dW_s
-
(K_T-K_t),
\end{equation}
with obstacle and Skorokhod conditions
\begin{equation}\label{eq:entropic-reflected-conditions}
V_t\le -\xi_t,
\qquad
\int_0^T (V_t+\xi_t)\,dK_t=0.
\end{equation}
Moreover, when the obstacle regularity ensures optimality of the first hitting
time, an optimal stopping time is
\[
\tau_t^*
=
\inf\{s\ge t:\;V_s=-\xi_s\}.
\]
Equivalently, this BSDE nonlinear expectation value has the explicit
reflected envelope representation
\begin{equation}\label{eq:reflected-entropic-envelope}
V_t
=
\essinf_{\tau\in\mathcal{T}_t}
\cE_{t,\tau}^g[-\xi_\tau]
=
\frac{1}{\gamma}
\log\left(
\essinf_{\tau\in\mathcal{T}_t}
\mathbb{E}\left[e^{-\gamma\xi_\tau}\mid\mathcal{F}_t\right]
\right).
\end{equation}
\end{proposition}

\begin{proof}
We first recall the non-reflected entropic calculation, since it is the basic
exponential transform behind the reflected formula. Fix a stopping time
$\tau\in\mathcal T_t$ and a terminal payoff $\xi_\tau$ with the required
exponential integrability. The quadratic BSDE on $[t,\tau]$ with terminal value
$-\xi_\tau$ is
\begin{equation}\label{eq:entropic-stopped-bsde}
Y_s^\tau
=
-\xi_\tau
+
\int_s^\tau \frac{\gamma}{2}|Z_r^\tau|^2\,dr
-
\int_s^\tau Z_r^\tau\,dW_r,
\qquad t\le s\le\tau .
\end{equation}
Set $M_s^\tau:=\exp(\gamma Y_s^\tau)$. By It\^o's formula,
\begin{align*}
dM_s^\tau
&=
\gamma M_s^\tau\,dY_s^\tau
+\frac{\gamma^2}{2}M_s^\tau |Z_s^\tau|^2\,ds  \\
&=
\gamma M_s^\tau
\left(-\frac{\gamma}{2}|Z_s^\tau|^2\,ds+Z_s^\tau\,dW_s\right)
+\frac{\gamma^2}{2}M_s^\tau |Z_s^\tau|^2\,ds  \\
&=
\gamma M_s^\tau Z_s^\tau\,dW_s .
\end{align*}
Thus $M^\tau$ is a martingale under the stated exponential moment assumptions,
and its terminal value is $M_\tau^\tau=e^{-\gamma\xi_\tau}$. Consequently
\begin{equation}\label{eq:entropic-explicit}
Y_t^\tau
=
\frac{1}{\gamma}
\log\mathbb{E}\left[e^{-\gamma\xi_\tau}\mid\mathcal F_t\right]
=
\cE_{t,\tau}^g[-\xi_\tau].
\end{equation}

Taking the essential infimum over all admissible stopping times gives
\begin{equation*}
V_t
=
\essinf_{\tau\in\mathcal T_t}Y_t^\tau
=
\essinf_{\tau\in\mathcal T_t}
\cE_{t,\tau}^g[-\xi_\tau]
=
\essinf_{\tau\in\mathcal T_t}
\frac{1}{\gamma}
\log\mathbb{E}\left[e^{-\gamma\xi_\tau}\mid\mathcal F_t\right].
\end{equation*}
Since $x\mapsto (1/\gamma)\log x$ is increasing, this is equivalent to
\eqref{eq:reflected-entropic-envelope}. If
\[
R_t:=\essinf_{\tau\in\mathcal T_t}
\mathbb{E}\left[e^{-\gamma\xi_\tau}\mid\mathcal F_t\right],
\]
then $R$ is the lower Snell envelope of the transformed payoff
$e^{-\gamma\xi}$. In particular, $R_t\le e^{-\gamma\xi_t}$, and applying the
increasing map $(1/\gamma)\log(\cdot)$ gives the upper obstacle condition
$V_t\le -\xi_t$.

We now identify explicitly how the reflection term transforms under the
exponential change of variables. Since
\[
R_t=e^{\gamma V_t},
\]
and the reflected quadratic BSDE can be written in differential form as
\[
dV_t
=
-\frac{\gamma}{2}|Z_t|^2\,dt
+
Z_t\,dW_t
+
dK_t,
\]
It\^o's formula gives
\begin{align*}
dR_t
&=
\gamma R_t\,dV_t
+
\frac{\gamma^2}{2}R_t|Z_t|^2\,dt\\
&=
\gamma R_t
\left(
-\frac{\gamma}{2}|Z_t|^2\,dt
+
Z_t\,dW_t
+
dK_t
\right)
+
\frac{\gamma^2}{2}R_t|Z_t|^2\,dt\\
&=
\gamma R_t Z_t\,dW_t
+
\gamma R_t\,dK_t.
\end{align*}
Thus, setting
\[
\widetilde Z_t:=\gamma R_tZ_t
\]
and
\[
A_t:=\int_0^t \gamma R_s\,dK_s,
\]
we obtain
\[
dR_t=\widetilde Z_t\,dW_t+dA_t.
\]
Since \(K\) is nondecreasing and \(R_t>0\), the process \(A\) is also
nondecreasing.

Moreover, the upper obstacle condition \(V_t\le-\xi_t\) is equivalent,
under the increasing exponential transformation, to
\[
R_t\le e^{-\gamma\xi_t}.
\]
The Skorokhod condition is preserved by the transformation. Indeed,
\(dK_t\) is carried by the contact set
\[
\{V_t=-\xi_t\},
\]
and therefore \(dA_t=\gamma R_t\,dK_t\) is carried by
\[
\{R_t=e^{-\gamma\xi_t}\}.
\]
Consequently,
\[
\int_0^T
\bigl(R_t-e^{-\gamma\xi_t}\bigr)\,dA_t
=
0.
\]
Hence \(R\) is precisely the lower Snell envelope of the transformed
payoff \(e^{-\gamma\xi}\), with Doob--Meyer decomposition
\[
dR_t=\widetilde Z_t\,dW_t+dA_t.
\]
Conversely, starting from the lower Snell envelope \(R\), define
\[
V_t:=\frac{1}{\gamma}\log R_t,
\qquad
Z_t:=\frac{\widetilde Z_t}{\gamma R_t},
\qquad
K_t:=\int_0^t\frac{1}{\gamma R_s}\,dA_s.
\]
Applying It\^o's formula to \(V_t=(1/\gamma)\log R_t\) yields
\[
dV_t
=
-\frac{\gamma}{2}|Z_t|^2\,dt
+
Z_t\,dW_t
+
dK_t.
\]
Together with
\[
V_t\le-\xi_t,
\qquad
\int_0^T(V_t+\xi_t)\,dK_t=0,
\]
this recovers the upper reflected quadratic BSDE
\[
V_t
=
-\xi_T
+
\int_t^T\frac{\gamma}{2}|Z_s|^2\,ds
-
\int_t^T Z_s\,dW_s
-
(K_T-K_t).
\]
Conversely, the comparison and minimality principle for reflected BSDEs
identify the first component of any solution with the essential infimum of the
entropic stopping evaluations. The first contact time
$\tau_t^*=\inf\{s\ge t:V_s=-\xi_s\}$ is optimal whenever the standard obstacle
regularity assumptions ensure that this hitting time is admissible and the
Snell envelope is attained there.
\end{proof}

The proposition above contains the detailed exponential transform proof. The
key point is that reflection does not destroy the entropic transform; it changes
the single conditional expectation in the non-reflected formula
\eqref{eq:entropic-explicit} into the lower Snell envelope appearing in
\eqref{eq:reflected-entropic-envelope}. Thus the remaining object is classical:
the optimal stopping envelope of the transformed payoff $e^{-\gamma\xi}$.

The proposition above reduces the entropic stopping problem to a classical
optimal stopping problem for the positive transformed payoff
\[
    U_t:=e^{-\gamma\xi_t}.
\]
In general, the lower Snell envelope of $U$ is the explicit object; a closed
formula requires additional structure on $\xi$. The following cases are useful
benchmarks for checking a numerical solver.

\begin{proposition}[Explicit stopping solutions for entropic risk]
\label{prop:explicit-entropic-stopping}
Let
\[
U_t:=e^{-\gamma\xi_t},
\qquad
R_t:=\operatorname*{ess\,inf}_{\tau\in\mathcal T_t}
\mathbb E[U_\tau\mid\mathcal F_t],
\qquad
V_t:=\frac{1}{\gamma}\log R_t .
\]
Then the following cases admit explicit optimal stopping rules.

\begin{enumerate}
\item[(i)]
If \(U\) is a submartingale, then
\[
R_t=U_t=e^{-\gamma\xi_t},
\qquad
V_t=-\xi_t,
\]
and immediate stopping
\[
\tau_t^*=t
\]
is optimal.

\item[(ii)]
If \(U\) is a supermartingale, then
\[
R_t
=
\mathbb E[U_T\mid\mathcal F_t]
=
\mathbb E[e^{-\gamma\xi_T}\mid\mathcal F_t],
\]
and therefore
\[
V_t
=
\frac{1}{\gamma}
\log
\mathbb E[e^{-\gamma\xi_T}\mid\mathcal F_t].
\]
In this case terminal stopping
\[
\tau_t^*=T
\]
is optimal.

\item[(iii)]
If \(\xi=(\xi_s)_{0\le s\le T}\) is deterministic and continuous, then
\[
R_t
=
\min_{s\in[t,T]}e^{-\gamma\xi_s}
=
e^{-\gamma\max_{s\in[t,T]}\xi_s},
\]
and hence
\[
V_t=-\max_{s\in[t,T]}\xi_s.
\]
An optimal stopping time is the first future maximizer
\[
\tau_t^*
=
\inf\left\{
s\in[t,T]:
\xi_s=\max_{u\in[t,T]}\xi_u
\right\}.
\]

\item[(iv)]
Suppose that
\[
\xi_t=x_0+\mu t+\sigma W_t,
\qquad \sigma\neq0.
\]
Then
\[
U_t
=
e^{-\gamma x_0}
\exp\bigl(-\gamma\mu t-\gamma\sigma W_t\bigr),
\]
and
\[
dU_t
=
U_t
\left[
\left(
-\gamma\mu+\frac{\gamma^2\sigma^2}{2}
\right)dt
-\gamma\sigma\,dW_t
\right].
\]
Consequently, if
\[
\mu<\frac{\gamma\sigma^2}{2},
\]
then \(U\) is a submartingale, immediate stopping is optimal, and
\[
R_t=e^{-\gamma\xi_t},
\qquad
V_t=-\xi_t.
\]

If
\[
\mu>\frac{\gamma\sigma^2}{2},
\]
then \(U\) is a supermartingale, terminal stopping is optimal, and
\[
R_t
=
\exp\left(
-\gamma\xi_t
+
\left[
-\gamma\mu+\frac{\gamma^2\sigma^2}{2}
\right](T-t)
\right),
\]
so that
\[
V_t
=
-\xi_t
+
\left(
-\mu+\frac{\gamma\sigma^2}{2}
\right)(T-t).
\]

At the critical value
\[
\mu=\frac{\gamma\sigma^2}{2},
\]
the process \(U\) is a martingale. Hence every stopping time
\(\tau\in\mathcal T_t\) gives the same value,
\[
R_t=U_t,
\qquad
V_t=-\xi_t.
\]
\end{enumerate}
\end{proposition}

\begin{proof}
For (i), if \(U\) is a submartingale, then for every
\(\tau\in\mathcal T_t\), the optional sampling theorem gives
\[
\mathbb E[U_\tau\mid\mathcal F_t]\ge U_t.
\]
Since \(\tau=t\) is admissible,
\[
R_t
=
\operatorname*{ess\,inf}_{\tau\in\mathcal T_t}
\mathbb E[U_\tau\mid\mathcal F_t]
=
U_t.
\]
Therefore
\[
V_t
=
\frac{1}{\gamma}\log U_t
=
-\xi_t,
\]
and immediate stopping is optimal.

For (ii), if \(U\) is a supermartingale, then for every
\(\tau\in\mathcal T_t\),
\[
\mathbb E[U_T\mid\mathcal F_\tau]\le U_\tau.
\]
Taking conditional expectations with respect to \(\mathcal F_t\) yields
\[
\mathbb E[U_T\mid\mathcal F_t]
\le
\mathbb E[U_\tau\mid\mathcal F_t].
\]
Since \(\tau=T\) is admissible,
\[
R_t=\mathbb E[U_T\mid\mathcal F_t],
\]
which gives the stated formula for \(V_t\).

For (iii), since \(\xi\) is deterministic,
\[
R_t
=
\min_{s\in[t,T]}e^{-\gamma\xi_s}.
\]
Because \(x\mapsto e^{-\gamma x}\) is decreasing,
\[
\min_{s\in[t,T]}e^{-\gamma\xi_s}
=
e^{-\gamma\max_{s\in[t,T]}\xi_s}.
\]
Continuity of \(\xi\) ensures that the maximum is attained, and therefore
the first maximizer is an optimal stopping time.

Finally, for (iv), It\^o's formula applied to
\(U_t=e^{-\gamma\xi_t}\) gives
\[
dU_t
=
U_t
\left[
\left(
-\gamma\mu+\frac{\gamma^2\sigma^2}{2}
\right)dt
-\gamma\sigma\,dW_t
\right].
\]
Thus \(U\) is a submartingale when
\[
\mu<\frac{\gamma\sigma^2}{2},
\]
a supermartingale when
\[
\mu>\frac{\gamma\sigma^2}{2},
\]
and a martingale at equality. The corresponding optimal stopping rules
follow from (i) and (ii).

In the supermartingale case, conditioning on the Brownian increment
\(W_T-W_t\) gives
\[
\mathbb E[e^{-\gamma\xi_T}\mid\mathcal F_t]
=
\exp\left(
-\gamma\xi_t
+
\left[
-\gamma\mu+\frac{\gamma^2\sigma^2}{2}
\right](T-t)
\right),
\]
and applying \((1/\gamma)\log(\cdot)\) yields
\[
V_t
=
-\xi_t
+
\left(
-\mu+\frac{\gamma\sigma^2}{2}
\right)(T-t).
\]
\end{proof}

\begin{remark}[Two Brownian benchmark examples]
\label{rem:brownian-entropic-benchmarks}
Proposition~\ref{prop:explicit-entropic-stopping}(iv) yields two particularly
simple benchmark examples.

First, consider
\[
\xi_t
=
aW_t-\frac{\gamma a^2}{2}t+ct,
\qquad c>0.
\]
Here
\[
\mu=c-\frac{\gamma a^2}{2},
\qquad
\sigma=a.
\]
Hence \(U_t=e^{-\gamma\xi_t}\) is a supermartingale precisely when
\[
c\ge\gamma a^2.
\]
If \(c>\gamma a^2\), terminal stopping is optimal and
\[
V_t
=
-aW_t
+\gamma a^2T
-\frac{\gamma a^2}{2}t
-cT.
\]
At the critical value \(c=\gamma a^2\), \(U\) is a martingale, so every
stopping time in \(\mathcal T_t\) is optimal and
\[
V_t=-\xi_t.
\]

Second, consider
\[
\xi_t
=
aW_t-\frac{\gamma a^2}{2}t-ct,
\qquad c>0.
\]
In this case
\[
\mu=-c-\frac{\gamma a^2}{2}
<
\frac{\gamma a^2}{2},
\]
so \(U\) is a submartingale. Consequently, immediate stopping is optimal,
\[
\tau_t^*=t,
\]
and
\[
V_t=-\xi_t.
\]
\end{remark}

\section{Algorithm and its Convergence}
\label{sec:numerical_convergence}

In this section we establish convergence of a deep learning scheme for the
reflected BSDE of Section~\ref{sec: assumpt}. Two obstructions have to be
removed. The driver is quadratic in $z$, whereas the available error estimates
for backward schemes require a globally Lipschitz driver; and the martingale
integrand $Z$ of a continuously reflected quadratic BSDE need not be bounded, so
the driver cannot simply be truncated without changing the solution. Both are
handled by passing through the discretely reflected equation, following
\cite{sun-liang-tang}, on which the integrand \emph{is} uniformly bounded and the
truncation is therefore exact. The resulting Lipschitz system is then treated by
the reflected deep backward scheme of \cite{hure-etal}.

\subsection{Markovian setting}\label{subsec:markov}

Throughout this section we work in a Markovian framework, which is more
restrictive than that of Section~\ref{sec: assumpt}. Let $X$ solve
\begin{equation}\label{eq:forward}
X_t=x_0+\int_0^t b(s,X_s)\,ds+\int_0^t\sigma(s)\,dW_s,
\qquad t\in[0,T],
\end{equation}
let the obstacle be $\xi_t=h(X_t)$ for a deterministic $h$, and let the driver be
$g(t,x,y,z)$. The upper reflected BSDE \eqref{eq:upper-rbsde} then reads
\begin{equation}\label{eq:rbsde-markov}
\begin{cases}
Y_t=-h(X_T)+\displaystyle\int_t^T g(s,X_s,Y_s,Z_s)\,ds-\int_t^T Z_s\,dW_s-(K_T-K_t),\\[4pt]
Y_t\le-h(X_t),\qquad \displaystyle\int_0^T\bigl(Y_t+h(X_t)\bigr)\,dK_t=0,\\[4pt]
K\ \text{continuous nondecreasing},\qquad K_0=0 .
\end{cases}
\end{equation}

\begin{enumerate}[label=(H\arabic*),leftmargin=*]
\item\label{assump:forward} \emph{Forward process.} $b$ and $\sigma$ are
deterministic, $\sigma$ does not depend on $x$, and for some $L>0$
\[
|b(t,0)|+|\sigma(t)|\le L,
\qquad
|b(t,x)-b(t,x')|\le L|x-x'| .
\]
\item\label{assump:obstacle} \emph{Obstacle.} $h:\R^d\to\R$ satisfies
$|h(x)|\le M_h$ and $|h(x)-h(x')|\le L|x-x'|$. No sign condition is imposed.
\item\label{assump:driver} \emph{Driver.} $g:[0,T]\times\R^d\times\R\times\R^d\to\R$
is continuous and, for constants $M_g,\gamma,L>0$,
\begin{align}
|g(t,x,y,z)|&\le M_g(1+|y|)+\tfrac{\gamma}{2}|z|^2, \label{eq:growth}\\
|g(t,x,y,z)-g(t,x',y,z)|&\le L\bigl(1+|z|\bigr)|x-x'|, \label{eq:lipx}\\
|g(t,x,y,z)-g(t,x,y',z)|&\le L|y-y'|, \label{eq:lipy}\\
|g(t,x,y,z)-g(t,x,y,z')|&\le L\bigl(1+|z|+|z'|\bigr)|z-z'| . \label{eq:lipz}
\end{align}
\item\label{assump:time} \emph{Time regularity.} For $s\le t$,
\[
|b(t,x)-b(s,x)|+|\sigma(t)-\sigma(s)|+|g(t,x,y,z)-g(s,x,y,z)|\le L(t-s)^{1/2}.
\]
\end{enumerate}

\noindent
Two further conditions are used only to quantify the discretisation error:
\begin{enumerate}[label=(R\arabic*),leftmargin=*]
\item\label{assump:R1} $h\in C_b^1$ with $L$-Lipschitz derivative.
\item\label{assump:R2} $h\in C_b^2$ with $L$-Lipschitz first and second
derivatives, and $\sigma$ is $L$-Lipschitz in $t$.
\end{enumerate}

\begin{remark}\label{rem:setting}
Relative to Section~\ref{sec: assumpt}, \eqref{eq:lipz} is the integrated form of
the hypothesis $|\partial_zg|\le C(1+|z|)$ imposed there and requires no
differentiability; \eqref{eq:growth} strengthens the linear-quadratic growth
bound by requiring the $y$-independent part to be bounded rather than merely in
$\mathcal H^2$. The genuine restriction is Markovianity, which is unavoidable:
the networks below take $X_{t_i}$ as input, so the scheme can only represent
$Y_t=u(t,X_t)$. That the drivers of Section~\ref{sec:drivers} satisfy
\ref{assump:driver}--\ref{assump:time} is verified in
Remark~\ref{rem:drivers-satisfy}.
\end{remark}

\noindent
It is convenient to record the reversal
\begin{equation}\label{eq:flip}
\widehat Y:=-Y,\qquad \widehat Z:=-Z,\qquad \widehat K:=K,\qquad
f(t,x,y,z):=-g(t,x,-y,-z),
\end{equation}
under which \eqref{eq:rbsde-markov} becomes the lower reflected equation
\[
\widehat Y_t=h(X_T)+\int_t^T f(s,X_s,\widehat Y_s,\widehat Z_s)\,ds
-\int_t^T\widehat Z_s\,dW_s+(\widehat K_T-\widehat K_t),
\qquad \widehat Y_t\ge h(X_t),
\]
with $\widehat K$ nondecreasing. Since $|f(t,x,y,z)|=|g(t,x,-y,-z)|$ and the
moduli in \ref{assump:driver} are invariant under $(y,z)\mapsto(-y,-z)$, $f$
satisfies \ref{assump:driver}--\ref{assump:time} with the same constants. The
results of \cite{sun-liang-tang} and \cite{hure-etal} are stated for the lower
reflected equation and are quoted below for $f$; all statements transfer to
\eqref{eq:rbsde-markov} through \eqref{eq:flip}.

\subsection{Discrete reflection and a priori bounds}\label{subsec:apriori}

Given a grid $\pi=\{0=t_0<\cdots<t_N=T\}$ with $|\pi|=\max_i(t_{i+1}-t_i)$ and
$N|\pi|\le L$, let $X^\pi$ denote the Euler scheme for \eqref{eq:forward}, so
that
\begin{equation}\label{eq:euler}
\E\Bigl[\sup_{0\le t\le T}|X_t-X^\pi_t|^{2p}\Bigr]\le C|\pi|^p,
\qquad p\ge1 ,
\end{equation}
which is the classical strong convergence estimate for the Euler--Maruyama
scheme under \ref{assump:forward}; see for instance \cite{kloeden-platen}.
For a set of reflection dates $R\subseteq\pi$ let $(Y^R,Z^R)$ denote the solution
of the discretely reflected equation, in which reflection acts only at the dates
of $R$ and the equation is solved exactly in between, and let
$(Y^{R,e},Z^{R,e})$ denote the same object driven by $X^\pi$ rather than $X$. We
write $(Y^e,Z^e,K^e)$ for the solution of \eqref{eq:rbsde-markov} driven by
$X^\pi$.

\begin{lemma}[\cite{sun-liang-tang}, Prop.~2.2 and Lem.~4.5]
\label{lem:a_priori_bounds}
Under \ref{assump:forward}--\ref{assump:driver} there are constants $M$ and
$M_z$, depending only on $L,M_g,M_h,\gamma$ and $T$, such that
\begin{equation}\label{eq:Ybound}
\|Y\|_\infty\le M,
\qquad
\|Z\cdot W\|_{\mathrm{BMO}}^2\le\frac{e^{4\gamma M}}{\gamma^2}
\bigl[1+2\gamma M_g(1+M)T\bigr],
\end{equation}
and, for every $R\subseteq\pi$,
\begin{equation}\label{eq:Zbound}
|Z^R_t|\le M_z,
\qquad
|Z^{R,e}_t|\le M_z,
\qquad 0\le t\le T .
\end{equation}
\end{lemma}

\begin{remark}\label{rem:ZnotZR}
The bound \eqref{eq:Zbound} holds for the \emph{discretely} reflected integrand,
uniformly in $R$, and not for $Z$ itself: the integrand of a continuously
reflected quadratic BSDE need not be bounded, and \eqref{eq:Ybound} gives only a
BMO estimate. This is the reason for introducing $(Y^R,Z^R)$ at all. The
hypothesis that $\sigma$ be deterministic in \ref{assump:forward} is what secures
\eqref{eq:Zbound}; the case of state-dependent $\sigma$ is open, see
\cite{sun-liang-tang}.
\end{remark}

\subsection{Truncated driver and the scheme}\label{subsec:scheme}

Let $\Pi_R(x)=\max(-R,\min(x,R))$ act componentwise and set
\begin{equation}\label{eq:truncated-driver-def}
g^{M_z}(t,x,y,z):=g\bigl(t,x,y,\Pi_{M_z}(z)\bigr),
\end{equation}
the truncation radius being fixed to the constant $M_z$ of
Lemma~\ref{lem:a_priori_bounds}. Since $|\Pi_{M_z}(z)|\le\sqrt d\,M_z$,
\eqref{eq:lipx}--\eqref{eq:lipz} give, for all $(x,y,z)$ and $(x',y',z')$,
\begin{equation}\label{eq:globlip}
\bigl|g^{M_z}(t,x,y,z)-g^{M_z}(t,x',y',z')\bigr|
\le L\bigl(1+\sqrt d\,M_z\bigr)|x-x'|+L|y-y'|+L\bigl(1+2\sqrt d\,M_z\bigr)|z-z'| ,
\end{equation}
so $g^{M_z}$ is globally Lipschitz. By \eqref{eq:Zbound},
$\Pi_{M_z}(Z^R_t)=Z^R_t$ and $\Pi_{M_z}(Z^{R,e}_t)=Z^{R,e}_t$, whence
\begin{equation}\label{eq:exact-truncation}
(Y^{R,M_z},Z^{R,M_z})=(Y^R,Z^R),
\qquad
(Y^{R,e,M_z},Z^{R,e,M_z})=(Y^{R,e},Z^{R,e}) :
\end{equation}
truncation changes nothing at the discretely reflected level. 

We now define the scheme, taking $R=\pi$. At each date two networks
$\mathcal U_i(\cdot;\theta)$ and $\mathcal Z_i(\cdot;\theta)$ approximate the
value and the integrand.

\begin{algorithm}[!htb]
\caption{Reflected deep backward dynamic programming with upper obstacle}
\label{alg:rdbdp}
\begin{enumerate}[label=\textup{Step \arabic*.}, leftmargin=*]
\item \textbf{Input.} Simulate the Euler paths
$(X^\pi_{t_i},\Delta W_i)_{i=0}^{N-1}$; the obstacle function $h$; the truncated
driver $g^{M_z}$ of \eqref{eq:truncated-driver-def}; network classes
$\mathcal U_i,\mathcal Z_i$.

\item \textbf{Terminal condition.} Set
$\widehat{\mathcal U}_N(X^\pi_{t_N})=-h(X^\pi_{t_N})$.

\item \textbf{Backward regression.} For $i=N-1,\ldots,0$, given
$\widehat{\mathcal U}_{i+1}$, minimise
\begin{align}\label{eq:section6-rdbdp-loss}
\widehat L_i(\theta)
:=\E\Bigl[\bigl|&\widehat{\mathcal U}_{i+1}(X^\pi_{t_{i+1}})
-\mathcal U_i(X^\pi_{t_i};\theta)\\
&+g^{M_z}\bigl(t_i,X^\pi_{t_i},\mathcal U_i(X^\pi_{t_i};\theta),
\mathcal Z_i(X^\pi_{t_i};\theta)\bigr)\Delta t_i
-\mathcal Z_i(X^\pi_{t_i};\theta)\Delta W_i\bigr|^2\Bigr],\nonumber
\end{align}
and let $\theta_i^*$ be a minimiser.

\item \textbf{Upper reflection.} Set
\begin{equation}\label{eq:hard_projection}
\widehat{\mathcal U}_i(x)=\min\bigl\{\mathcal U_i(x;\theta_i^*),\,-h(x)\bigr\},
\qquad
\widehat{\mathcal Z}_i(x)=\mathcal Z_i(x;\theta_i^*).
\end{equation}

\item \textbf{Output.} $(\widehat{\mathcal U}_i,\widehat{\mathcal Z}_i)_{i=0}^{N-1}$
and $\widehat Y_0=\widehat{\mathcal U}_0(x_0)$.
\end{enumerate}
\end{algorithm}

\noindent
Algorithm~\ref{alg:rdbdp} projects at every date, so it is discretely reflected
with $R=\pi$, the case of \eqref{eq:exact-truncation}. The estimates of
\cite{sun-liang-tang} below hold for general $R\subseteq\pi$; taking $R=\pi$
replaces their $|R|$ by $|\pi|$ and their $\alpha_i(\kappa)$ by $\alpha_i(N)$.

\subsection{Convergence}\label{subsec:convergence}

Write $(\bar Y^\pi,\bar Z^\pi)$ for the discrete-time scheme obtained from
Algorithm~\ref{alg:rdbdp} when the conditional expectations are computed exactly,
and let
\[
\epsilon_i^{\mathcal N,\tilde u}:=\inf_\xi\E\bigl|\tilde u_i(X^\pi_{t_i})-\mathcal U_i(X^\pi_{t_i};\xi)\bigr|^2,
\qquad
\epsilon_i^{\mathcal N,\tilde z}:=\inf_\eta\E\bigl|\tilde z_i(X^\pi_{t_i})-\mathcal Z_i(X^\pi_{t_i};\eta)\bigr|^2
\]
be the $L^2$ approximation errors of the network classes for the conditional
expectations $\tilde u_i,\tilde z_i$ arising at step $i$. The total error is
measured by
\begin{equation}\label{eq:error-functional}
\cE\bigl[(\widehat{\mathcal U},\widehat{\mathcal Z}),(Y,Z)\bigr]
:=\max_{i=0,\dots,N-1}\E\bigl|Y_{t_i}-\widehat{\mathcal U}_i(X^\pi_{t_i})\bigr|^2
+\E\Bigl[\sum_{i=0}^{N-1}\int_{t_i}^{t_{i+1}}
\bigl|Z_t-\widehat{\mathcal Z}_i(X^\pi_{t_i})\bigr|^2dt\Bigr].
\end{equation}

\begin{proposition}\label{thm:main_convergence}
Let \ref{assump:forward}--\ref{assump:time} hold and let $R=\pi$. There is a
constant $C>0$, independent of $\pi$ and of the network classes, such that
\begin{equation}\label{eq:main-bound}
\cE\bigl[(\widehat{\mathcal U},\widehat{\mathcal Z}),(Y,Z)\bigr]
\le C\Bigl(\varepsilon(\pi)
+\sum_{i=0}^{N-1}\bigl(N\epsilon_i^{\mathcal N,\tilde u}+\epsilon_i^{\mathcal N,\tilde z}\bigr)\Bigr),
\end{equation}
where $\varepsilon(\pi)$ is the discretisation error. Under \ref{assump:R1},
$\varepsilon(\pi)=O(|\pi|^{1/2})$ for the $Y$-component and
$O(|\pi|^{1/2})$ for the $Z$-component of \eqref{eq:error-functional}; under
\ref{assump:R2} the $Y$-component improves to $O(|\pi|)$, while the
$Z$-component does not improve.

Here these rates refer to the squared error quantities in
\eqref{eq:error-functional}. Equivalently, at the norm level, the
corresponding rate for the $Y$-component is $O(|\pi|^{1/4})$ under
\ref{assump:R1} and $O(|\pi|^{1/2})$ under \ref{assump:R2}, while the
$H^2$ error for the $Z$-component is $O(|\pi|^{1/4})$ under both
assumptions.
\end{proposition}

\begin{proof}
Throughout, constants depend only on the data of
\ref{assump:forward}--\ref{assump:time} and on $M_z$. By
\eqref{eq:exact-truncation} the driver may be replaced by $g^{M_z}$ in every
discretely reflected object without altering it, and by \eqref{eq:globlip} that
driver is globally Lipschitz, so all estimates below are for the Lipschitz
system. The estimates for the passage to discrete reflection and for the time
discretisation are those of \cite{sun-liang-tang}, and the estimate for the
network error is that of \cite{hure-etal}; both are restated here for
completeness and are not original to this paper. What is established below is
their combination.

For the value component, introduce the discrete $\mathcal S^2$ norm
\[
\|U\|_{\mathcal S^2_\pi}
:=
\left(
\max_{0\le i\le N-1}\mathbb E|U_{t_i}|^2
\right)^{1/2}.
\]
Then
\[
\|Y-\widehat {\mathcal U} \|_{\mathcal S^2_\pi}
\le
\|Y-Y^R\|_{\mathcal S^2_\pi}
+
\|Y^R-\bar Y^\pi\|_{\mathcal S^2_\pi}
+
\|\bar Y^\pi-\widehat {\mathcal U}\|_{\mathcal S^2_\pi}.
\]
The three terms are respectively the passage from continuous to discrete
reflection, the time discretisation, and the replacement of the conditional
expectations by networks. For the first, \cite[Lem.~4.2]{sun-liang-tang} gives
\[
\max_{j}\Bigl\|\sup_{t\in[r_j,r_{j+1}]}|Y_t-Y^{R}_t|\Bigr\|_{L^2}
\le C|\pi|^{1/4},
\]
improving to $C|\pi|^{1/2}$ under \ref{assump:R1}. For the second, the system
being Lipschitz, \cite[Lem.~4.1]{sun-liang-tang} gives
\[
\bigl\|Y^{R}-\bar Y^\pi\bigr\|_{L^2}
\le C\bigl(\alpha_1(N)|\pi|^{1/2}+\epsilon_1(\pi)\bigr),
\qquad
(\alpha_1,\epsilon_1)=
\begin{cases}
(N^{1/4},|\pi|^{1/4}) & \text{under \ref{assump:R1}},\\
(1,|\pi|^{1/2}) & \text{under \ref{assump:R2}}.
\end{cases}
\]
Since $N|\pi|\le L$,
\[
\alpha_1(N)|\pi|^{1/2}=
\begin{cases}
N^{1/4}|\pi|^{1/2}\le L^{1/4}|\pi|^{1/4} & \text{under \ref{assump:R1}},\\
|\pi|^{1/2} & \text{under \ref{assump:R2}},
\end{cases}
\]
and the first two terms together are $O(|\pi|^{1/4})$ under \ref{assump:R1} and
$O(|\pi|^{1/2})$ under \ref{assump:R2}.

For the martingale integrand the same route is unavailable, the corresponding
estimate in \cite[Lem.~4.1]{sun-liang-tang} carrying $\alpha_2(N)=N^{1/2}$
rather than $\alpha_1$, so that $\alpha_2(N)|\pi|^{1/2}=(N|\pi|)^{1/2}$ does not
vanish. One passes instead through the family driven by $X^\pi$. The stability
estimate \cite[Thm.~3.2]{sun-liang-tang} together with \eqref{eq:euler} gives
\[
\bigl\|Z-Z^e\bigr\|^2_{\mathcal H^2}\le C\bigl\|X-X^\pi\bigr\|\le C|\pi|^{1/2},
\]
then $\|Z^e-Z^{R,e}\|_{\mathcal H^2}\le C|\pi|^{1/4}$ by
\cite[Lem.~4.2]{sun-liang-tang}, and finally the third estimate of
\cite[Lem.~4.1]{sun-liang-tang} gives
\[
\bigl\|Z^{R,e}-\bar Z^\pi\bigr\|_{\mathcal H^2}
\le C\bigl(\alpha_1(N)|\pi|^{1/2}+\epsilon_2(\pi)\bigr),
\qquad \epsilon_2(\pi)=|\pi|^{1/4},
\]
with $\alpha_1$ as above. Each of these three terms is $O(|\pi|^{1/4})$, and $\epsilon_2$ is unchanged by \ref{assump:R2}, which is the
assertion that the estimate for the integrand does not improve under the
stronger regularity.

It remains to bound the network error, the third term of the decomposition
above. By \eqref{eq:globlip} the system driven by $g^{M_z}$ satisfies hypothesis
(H1) of \cite{hure-etal}, and Algorithm~\ref{alg:rdbdp} is their RDBDP scheme
with $\min$ in place of $\max$ in the projection step
\eqref{eq:hard_projection}. The argument of \cite[proof of Thm.~4.4, from
(4.41)]{hure-etal} therefore applies verbatim, the only change being the use of
\[
\bigl|\min(a,c)-\min(b,c)\bigr|\le|a-b|,
\qquad a,b,c\in\R,
\]
in place of the corresponding inequality for $\max$. That argument uses only the
Lipschitz property of the driver and the identity
$\arg\min_\theta\widehat L_i(\theta)=\arg\min_\theta\widetilde L_i(\theta)$, and
is insensitive to the law of the forward process, so $X$ may be replaced by
$X^\pi$ throughout. It yields
\[
\max_{i=0,\dots,N-1}\E\bigl|\bar Y^\pi_{t_i}-\widehat{\mathcal U}_i(X^\pi_{t_i})\bigr|^2
\le C\sum_{i=0}^{N-1}\bigl(N\epsilon_i^{\mathcal N,\tilde u}+\epsilon_i^{\mathcal N,\tilde z}\bigr),
\]
and the corresponding bound for the integrand by the argument of
\cite[proof of Thm.~4.1, Step 5]{hure-etal}. Combining the displays gives
\eqref{eq:main-bound}.
\end{proof}

\begin{remark}\label{rem:constant}
The constant $C$ in \eqref{eq:main-bound} depends on the Lipschitz constant of
$g^{M_z}$, hence by \eqref{eq:globlip} on $M_z$, which involves exponential terms
in the data. The estimate is therefore qualitatively sharp but numerically weak.
An alternative, developed for unreflected quadratic BSDEs in
\cite{chassagneux-richou}, is to let the truncation radius grow with the grid
rather than fixing it, balancing a nonzero truncation error against a slowly
growing Lipschitz constant and recovering a rate $|\pi|^{1-\eta}$ for every
$\eta>0$ with a controlled constant. Extending that construction to the reflected
deep backward scheme is left for future work, as is the case of state-dependent
$\sigma$ (Remark~\ref{rem:ZnotZR}).
\end{remark}

\begin{remark}\label{rem:provenance}
The discretisation estimates are obtained in \cite{sun-liang-tang} for a scheme
in which the conditional expectations are computed exactly, and the network error
estimate in \cite{hure-etal} for a Lipschitz driver. The contribution of this
section is the combination of the two, the treatment of the upper obstacle, and
the application to the drivers of Section~\ref{sec:drivers}.
\end{remark}

\section{Examples}\label{sec:drivers}

The purpose of this section is to illustrate how different specifications of the reflected 
BSDE driver correspond to different financial interpretations of ambiguity and
risk evaluation. In the reflected setting, the driver determines the nonlinear
risk assessment during the continuation region, while the obstacle represents
the immediate liquidation or exercise value. Thus each example below should be
read as a particular financial mechanism generating a dynamic stopping risk
measure.

From a risk measure perspective, the variable \(Y\) represents the current
risk adjusted value of the future position, whereas \(Z\) captures the exposure
of this value to Brownian market uncertainty. A linear term in \(Y\) naturally
corresponds to discounting, or equivalently to the time value of future cash
flows. A quadratic term in \(Z\) corresponds to entropic model ambiguity,
exponential utility valuation, or robustness with respect to changes of
probability measure.

The first example isolates ambiguity in the discount factor and is related to
cash-subadditive risk measures. The second example recalls the classical
discounted entropic driver. The third example combines both mechanisms by
allowing uncertainty in the discount rate and in the entropic ambiguity
parameter. Finally, the American put example shows how these drivers enter a
reflected BSDE with an economically meaningful stopping boundary.

\subsection{Bounded discount rate ambiguity}\label{subsec:bounded-beta}

Let the discount rate belong to the bounded interval
\[
\beta_t\in[\underline\beta,\overline\beta],
\qquad
0\le \underline\beta\le \overline\beta .
\]
For \(t\le \tau\), define
\[
\rho_{t,\tau}(\xi_\tau)
:=
\sup_{\beta\in[\underline\beta,\overline\beta]}
\mathbb E
\left[
e^{-\int_t^\tau \beta_s\,ds}
(-\xi_\tau)
\;\middle|\;
\mathcal F_t
\right].
\]
This is a pure discount-ambiguity risk measure. The probability measure is fixed,
and the adversary only chooses the discount rate.

The associated reflected BSDE driver is
\[
g(t,y,z)
=
\sup_{\beta\in[\underline\beta,\overline\beta]}
\{-\beta y\}.
\]
For simplicity and to obtain an explicit benchmark, we take the ambiguity
bounds $\underline\beta$ and $\overline\beta$ to be constant. More generally,
one may consider bounded progressively measurable processes
$\underline\beta_t$ and $\overline\beta_t$, with
$0\leq\underline\beta_t\leq\overline\beta_t$, without changing the basic
supremum structure of the driver.
Equivalently,
\[
g(t,y,z)
=
-\overline\beta y\,\mathbf 1_{\{y<0\}}
-
\underline\beta y\,\mathbf 1_{\{y\ge0\}}.
\]

The stopping value
\[
Y_t
=
\essinf_{\tau\in\mathcal T_t}
\rho_{t,\tau}(\xi_\tau)
\]
is represented by the upper reflected BSDE
\[
Y_t
=
-\xi_T
+
\int_t^T
\sup_{\beta\in[\underline\beta,\overline\beta]}
\{-\beta Y_s\}\,ds
-
\int_t^T Z_s\,dW_s
-
(K_T-K_t),
\]
with
\[
Y_t\le -\xi_t,
\qquad
\int_0^T(Y_t+\xi_t)\,dK_t=0.
\]
The optimal stopping time is
\[
\tau_t^*
=
\inf\{s\ge t:Y_s=-\xi_s\}.
\]


In the liability region, where \(Y_s\le0\), the worst case discount rate is
\(\overline\beta\). Hence the equation reduces locally to
\[
Y_t
=
-\xi_T
+
\int_t^T
(-\overline\beta Y_s)\,ds
-
\int_t^T Z_s\,dW_s
-
(K_T-K_t).
\]

If the payoff process is deterministic, continuous, and nonnegative, then
\[
\rho_{t,s}(\xi_s)
=
-e^{-\overline\beta(s-t)}\xi_s,
\qquad s\in[t,T],
\]
and therefore
\[
Y_t
=
-\sup_{s\in[t,T]}
\left\{
e^{-\overline\beta(s-t)}\xi_s
\right\}.
\]
An optimal stopping time is any first maximizer of the discounted payoff:
\[
\tau_t^*
=
\inf\left\{
s\in[t,T]:
e^{-\overline\beta(s-t)}\xi_s
=
\sup_{u\in[t,T]}
e^{-\overline\beta(u-t)}\xi_u
\right\}.
\]

\subsection{Linear discounting and change of measure}
\label{subsec:linear-discount-girsanov}

We next consider an affine driver combining deterministic discounting with
a linear change-of-measure component. Let
\[
g(t,y,z)
=
-\beta y+\gamma\cdot z,
\qquad
\beta\geq 0,
\quad
\gamma\in\mathbb R^d.
\]
For simplicity, $\beta$ and $\gamma$ are taken to be constant in this
benchmark. More generally, the same representation extends to suitable
bounded progressively measurable processes $\beta_t$ and $\gamma_t$,
provided the corresponding stochastic exponential is a true martingale.

Since the driver is affine in $(y,z)$, Theorem~\ref{prop: properties}-(e)
implies that the associated stopping operator $\varphi_t$ is concave.

To obtain an explicit representation of the corresponding dynamic risk
evaluation, define the probability measure $\mathbb Q^\gamma$ by
\[
\left.
\frac{d\mathbb Q^\gamma}{d\mathbb P}
\right|_{\mathcal F_s}
=
\mathcal E\left(
\int_0^\cdot \gamma\cdot dW_u
\right)_s
=
\exp\left\{
\int_0^s \gamma\cdot dW_u
-\frac12\int_0^s|\gamma|^2\,du
\right\}.
\]
Since $\gamma$ is constant, Novikov's condition is automatically satisfied,
and hence the stochastic exponential is a true martingale. By Girsanov's
theorem, the process
\[
W_s^\gamma
:=
W_s-\int_0^s\gamma\,du
=
W_s-\gamma s
\]
is a $d$-dimensional Brownian motion under $\mathbb Q^\gamma$.

Indeed, for a fixed horizon $s\in[t,T]$, the BSDE associated with the
dynamic risk evaluation is
\[
Y_r
=
-\xi_s
+
\int_r^s
\bigl(
-\beta Y_u+\gamma\cdot Z_u
\bigr)\,du
-
\int_r^s Z_u\,dW_u,
\qquad r\in[t,s].
\]
In differential form,
\[
dY_r
=
\beta Y_r\,dr
-\gamma\cdot Z_r\,dr
+
Z_r\,dW_r.
\]
Since
\[
dW_r=dW_r^\gamma+\gamma\,dr,
\]
the terms involving $\gamma\cdot Z_r$ cancel, and therefore
\[
dY_r
=
\beta Y_r\,dr
+
Z_r\,dW_r^\gamma.
\]
Multiplying by the discount factor $e^{-\beta(r-t)}$ gives
\[
d\left(
e^{-\beta(r-t)}Y_r
\right)
=
e^{-\beta(r-t)}Z_r\,dW_r^\gamma.
\]
Hence the discounted process is a $\mathbb Q^\gamma$-martingale. Using the
terminal condition $Y_s=-\xi_s$, we obtain
\[
\rho_{t,s}(\xi_s)
=
Y_t
=
\mathbb E^{\mathbb Q^\gamma}
\left[
e^{-\beta(s-t)}(-\xi_s)
\,\middle|\,
\mathcal F_t
\right],
\qquad
\xi_s\in L^\infty(\mathcal F_s).
\]

Consequently, the corresponding optimal stopping risk functional admits
the representation
\[
\varphi_t(\xi)
=
\operatorname*{ess\,inf}_{\tau\in\mathcal T_t}
\rho_{t,\tau}(\xi_\tau)
=
\operatorname*{ess\,inf}_{\tau\in\mathcal T_t}
\mathbb E^{\mathbb Q^\gamma}
\left[
e^{-\beta(\tau-t)}(-\xi_\tau)
\,\middle|\,
\mathcal F_t
\right].
\]
Thus, in this affine case, the nonlinear stopping problem reduces to a
classical optimal stopping problem under the equivalent probability measure
$\mathbb Q^\gamma$, with discount rate $\beta$.
\subsection{Linear discounting with entropic model ambiguity}

A classical financially meaningful driver is
\[
g(t,y,z)
=
-\delta y+\frac{\gamma}{2}|z|^2,
\qquad
\delta\ge0,\quad \gamma>0.
\]
The term \(-\delta y\) represents linear discounting, while
\((\gamma/2)|z|^2\) is the standard entropic or model-ambiguity term.

The convex conjugate is computed as
\[
g^*(t,\beta,\alpha)
=
\sup_{(y,z) \in \R \times \R^d}
\left\{
-\beta y-\alpha\cdot z-g(t,y,z)
\right\}.
\]
Hence
\[
g^*(t,\beta,\alpha)
=
\sup_{y \in \R}\{(-\beta+\delta)y\}
+
\sup_{z \in \R^d}
\left\{
-\alpha\cdot z-\frac{\gamma}{2}|z|^2
\right\}.
\]
The supremum in \(z\) is attained at
\[
z^*=-\frac{\alpha}{\gamma},
\]
and gives
\[
\sup_{z \in \R^d}
\left\{
-\alpha\cdot z-\frac{\gamma}{2}|z|^2
\right\}
=
\frac{1}{2\gamma}|\alpha|^2.
\]
Therefore
\[
g^*(t,\beta,\alpha)
=
\begin{cases}
\dfrac{1}{2\gamma}|\alpha|^2,
& \beta=\delta,\\[1em]
+\infty,
& \beta\neq\delta.
\end{cases}
\]
Thus this model has entropic ambiguity in the probability measure, while the
discount rate is fixed at \(\delta\). It serves as the benchmark case for the
paired ambiguity model below.

\subsection{Entropic model ambiguity with uncertain discounting}

We now allow simultaneous ambiguity in the discount rate and in the entropic
ambiguity parameter. Let
\[
g(t,y,z)
=
\sup_{(\delta,\gamma)\in C}
\left\{
-\delta y+\frac{\gamma}{2}|z|^2
\right\},
\]
where \(\delta\) represents the discount rate and \(\gamma\) represents the
strength of entropic model ambiguity.

A separable ambiguity set is
\[
C=
[\underline\delta,\overline\delta]
\times
[\underline\gamma,\overline\gamma],
\]
with
\[
0\le \underline\delta\le\overline\delta,
\qquad
0<\underline\gamma\le\overline\gamma.
\]
In this case,
\[
g(t,y,z)
=
\sup_{\delta\in[\underline\delta,\overline\delta]}
(-\delta y)
+
\sup_{\gamma\in[\underline\gamma,\overline\gamma]}
\frac{\gamma}{2}|z|^2.
\]
Therefore
\[
g(t,y,z)
=
-\overline\delta y\,\mathbf 1_{\{y<0\}}
-
\underline\delta y\,\mathbf 1_{\{y\ge0\}}
+
\frac{\overline\gamma}{2}|z|^2.
\]
In the liability region \(y\le0\), this reduces to
\[
g(t,y,z)
=
-\overline\delta y
+
\frac{\overline\gamma}{2}|z|^2.
\]

For simplicity, we take the ambiguity bounds
\(\underline{\delta},\overline{\delta},
\underline{\gamma},\overline{\gamma}\)
to be constant in this benchmark. More generally, one may consider
bounded progressively measurable bounds
\[
\underline{\delta}_t\leq\overline{\delta}_t,
\qquad
\underline{\gamma}_t\leq\overline{\gamma}_t,
\]
leading to the time-dependent ambiguity set
\[
\mathcal C_t
=
[\underline{\delta}_t,\overline{\delta}_t]
\times
[\underline{\gamma}_t,\overline{\gamma}_t].
\]
The constant rectangular specification considered here is sufficient to
illustrate explicitly the interaction between discount rate ambiguity and
entropic model ambiguity.

We deliberately use a rectangular ambiguity set, so that the two sources
of ambiguity can vary independently. Indeed, the discount parameter
\(\delta\) and the entropic parameter \(\gamma\) play different roles in
the driver: \(\delta\) acts on the value variable \(y\), whereas
\(\gamma\) controls the quadratic exposure term \(|z|^2\).
For this reason, no direct ordering between \(\delta\) and \(\gamma\) is
imposed.

For fixed \((y,z)\), the worst-case parameters can be identified
explicitly. Since
\[
\frac{\gamma}{2}|z|^2
\]
is nondecreasing in \(\gamma\), the supremum over the entropic ambiguity
parameter is attained at
\[
\gamma^*=\overline{\gamma}.
\]
Similarly,
\[
\delta^*(y)
=
\begin{cases}
\overline{\delta}, & y<0,\\[1mm]
\underline{\delta}, & y\geq0.
\end{cases}
\]
Hence, in the liability region \(y\leq0\), the worst-case specification is
\[
(\delta^*,\gamma^*)
=
(\overline{\delta},\overline{\gamma}),
\]
and the driver becomes
\[
g(t,y,z)
=
-\overline{\delta}y
+
\frac{\overline{\gamma}}{2}|z|^2.
\]

The associated stopping value
\[
Y_t
=
\essinf_{\tau\in\mathcal T_t}
\rho_{t,\tau}(\xi_\tau)
\]
is represented by the upper reflected BSDE
\[
Y_t
=
-\xi_T
+
\int_t^T
\sup_{(\delta,\gamma)\in C}
\left\{
-\delta Y_s+\frac{\gamma}{2}|Z_s|^2
\right\}ds
-
\int_t^T Z_s\,dW_s
-
(K_T-K_t),
\]
with
\[
Y_t\le -\xi_t,
\qquad
\int_0^T(Y_t+\xi_t)\,dK_t=0.
\]
The optimal stopping time is
\[
\tau_t^*
=
\inf\{s\ge t:Y_s=-\xi_s\}.
\]

If the payoff is deterministic and nonnegative, then \(Z\equiv0\). In the
liability region, the driver becomes
\[
g(t,y,0)=-\overline\delta y.
\]
Consequently,
\[
\rho_{t,s}(\xi_s)
=
-e^{-\overline\delta(s-t)}\xi_s,
\]
and the stopping value is explicitly
\[
Y_t
=
-\sup_{s\in[t,T]}
\left\{
e^{-\overline\delta(s-t)}\xi_s
\right\}.
\]
Thus, in the deterministic case, the optimal stopping rule is to stop at the
first time at which the worst case discounted payoff is maximal.

\subsection{Geometric American put under entropic-discount ambiguity}
\label{subsec:put}

We consider an American put written on a single asset whose risk neutral price
\(S\) follows a geometric Brownian motion. In log-price coordinates
\(X_t=\log S_t\), the dynamics are
\[
dX_t
=
\left(r-\frac12\sigma^2\right)dt+\sigma\,dW_t,
\qquad
X_0=\log S_0.
\]
The exercise payoff is
\[
\xi_t=(K-S_t)^+=(K-e^{X_t})^+.
\]
It is bounded by \(K\) and has continuous paths, hence it satisfies the standing
obstacle assumptions.

Using the entropic discount ambiguity driver
\[
g(t,y,z)
=
\sup_{(\delta,\gamma)\in C}
\left\{
-\delta y+\frac{\gamma}{2}|z|^2
\right\},
\]
the stopping value
\[
Y_t
=
\essinf_{\tau\in\mathcal T_t}
\rho_{t,\tau}(\xi_\tau)
\]
solves the upper reflected BSDE
\[
Y_t
=
-\xi_T
+
\int_t^T
g(s,Y_s,Z_s)\,ds
-
\int_t^T Z_s\,dW_s
-
(K_T-K_t),
\qquad
Y_t\le -\xi_t,
\]
with
\[
\int_0^T(Y_t+\xi_t)\,dK_t=0.
\]

For the separable ambiguity set
\[
C=
[\underline\delta,\overline\delta]
\times
[\underline\gamma,\overline\gamma],
\]
and in the liability region \(Y_t\le0\), the driver reduces to
\[
g(t,Y_t,Z_t)
=
-\overline\delta Y_t
+
\frac{\overline\gamma}{2}|Z_t|^2.
\]
Thus the reflected BSDE used for the American put becomes
\[
Y_t
=
-\xi_T
+
\int_t^T
\left(
-\overline\delta Y_s
+
\frac{\overline\gamma}{2}|Z_s|^2
\right)ds
-
\int_t^T Z_s\,dW_s
-
(K_T-K_t),
\qquad
Y_t\le -\xi_t.
\]

The initial asset level and strike may be chosen as \(S_0=1\) and \(K=1.1\).
Then the option starts slightly in the money and the initial upper obstacle is
\[
-\xi_0=-(K-S_0)^+=-0.1.
\]
This makes the stopping constraint active and visible in the numerical
diagnostics.

\section{Numerical Diagnostics and Results}\label{sec:numerics}

We use the reflected deep backward dynamic programming scheme of
Section~\ref{sec:numerical_convergence}, Algorithm~\ref{alg:rdbdp}, to solve the
upper reflected BSDE \eqref{eq:upper-rbsde}--\eqref{eq:upper-skorokhod}
on three test cases:
\begin{enumerate}[label=\textup{(\arabic*)},leftmargin=*]
    \item American put under Black--Scholes ($g(y)=-r\,y$, no ambiguity),
    the standard optimal stopping benchmark for which a binomial reference price
    is available.
    \item Bounded discount rate ambiguity of Section~\ref{subsec:bounded-beta}, a
    standard linear example under ambiguity that still falls within the
    Lipschitz framework.
    \item Geometric American put under entropic discount ambiguity of
    Section~\ref{subsec:put}, which introduces the superlinear (quadratic in $z$)
    driver and is the non-Lipschitz target of the convergence analysis in
    Section~\ref{sec:numerical_convergence}.
\end{enumerate}

\begin{remark}\label{rem:drivers-satisfy}
The drivers of Section~\ref{sec:drivers} satisfy
\ref{assump:driver}--\ref{assump:time}. They are of the form
\[
g(t,x,y,z)=\sup_{(\delta,\gamma)\in\C}\psi(\delta,\gamma,y,z),
\]
with $\C$ bounded, $\delta\in[\underline\delta,\overline\delta]$ and
$\gamma\in[\underline\gamma,\overline\gamma]$, and are independent of $(t,x)$;
hence \eqref{eq:lipx} holds with $L=0$ and \ref{assump:time} is vacuous for $g$.
For $\psi(\delta,\gamma,y,z)=-\delta y+\tfrac{\gamma}{2}|z|^2$,
\[
-\overline\delta|y|
\;\le\;g(y,z)\;\le\;\overline\delta|y|+\tfrac{\overline\gamma}{2}|z|^2,
\]
giving \eqref{eq:growth} with $M_g=\overline\delta$ and
$\gamma=\overline\gamma$, while $|\sup_a F(a,\cdot)-\sup_a F(a,\cdot)|\le\sup_a|F(a,\cdot)-F(a,\cdot)|$
gives
\small
\[
|g(y,z)-g(y',z)|\le\sup_{\delta}\delta|y-y'|\le\overline\delta|y-y'|,
|g(y,z)-g(y,z')|\le\sup_\gamma\tfrac{\gamma}{2}\bigl||z|^2-|z'|^2\bigr|
\le\tfrac{\overline\gamma}{2}\bigl(|z|+|z'|\bigr)|z-z'| ,
\]
which are \eqref{eq:lipy} and \eqref{eq:lipz}. The remaining drivers of
Section~\ref{sec:drivers} are special cases. Boundedness of $\C$ is what these
estimates require.
\end{remark}
Cases (1) and (3) use the payoff $h(x)=(K-e^{x})^{+}$, which is bounded and
Lipschitz but not differentiable at the strike. They therefore satisfy the
standing assumptions \ref{assump:forward}--\ref{assump:time} of
Section~\ref{sec:numerical_convergence}, but not the additional regularity
\ref{assump:R1}--\ref{assump:R2} under which the discretisation error in
Proposition~\ref{thm:main_convergence} is quantified. Case (2) satisfies all of
them, and is the example for which the rate of
Proposition~\ref{thm:main_convergence} is available.
For each case we report (i) sample paths of the value process $Y_t$ against the
obstacle $-\xi_t$ and of the control process $Z_t$, (ii) the empirical
distribution of optimal stopping times, and (iii) the training loss at the
terminal time step. Numerical verification of the relevant properties of
Theorem~\ref{prop: properties} is collected in
Section~\ref{sec:numerics:properties}.

\paragraph{Common solver configuration.}
All experiments use the same network and training hyperparameters
(Table~\ref{tab:nn-params}). Outputs go through the upper-reflection projection
\eqref{eq:hard_projection}. Implementation is in Python with PyTorch on CPU.

\begin{table}[!htb]
\centering
\begin{tabular}{l l}
\hline
Architecture                           & 3 hidden layers, $50$ neurons per layer \\
Activation                             & $\tanh$ on hidden layers, identity on output \\
Optimizer                              & Adam, learning rate $10^{-3}$ \\
Batch size                             & $M = 2^{10} = 1024$ \\
Time discretization                    & $N = 50$ steps, $\Delta t = T/N$ \\
Base iterations per step               & $\mathrm{itr} = 300$ \\
Iteration multiplier on last $3$ steps & $\times 10$ \\
\hline
\end{tabular}
\caption{Network and training hyperparameters used in all three cases.}
\label{tab:nn-params}
\end{table}

\subsection{American put under Black--Scholes}
\label{subsec:num-bs}

This case is the canonical optimal stopping benchmark: the reflected BSDE of
Theorem~\ref{thm:rbsde-representation} reduces, under the affine driver
$g(y) = -r y$ and the put payoff $\xi_t = (K-S_t)^+$, to the classical American
put under Black--Scholes. The Cox--Ross--Rubinstein binomial price is available
as a reference, which fixes the baseline accuracy of the scheme before any
ambiguity is introduced.

The forward process is the GBM in log coordinates of Section~\ref{subsec:put}; the
upper obstacle is $-\xi_t = -(K-e^{X_t})^+$. The driver is recovered from
Section~\ref{sec:drivers} as the discount-only special case of
Section~\ref{subsec:bounded-beta} with $\underline\beta = \overline\beta = r$. Market
parameters are listed in Table~\ref{tab:bs-params}.

\begin{table}[!htb]
\centering
\begin{tabular}{l l l}
\hline
Initial price        & $S_0$    & $1.0$ \\
Strike               & $K$      & $1.1$ \\
Volatility           & $\sigma$ & $0.2$ \\
Risk-free rate       & $r$      & $0.05$ \\
Time horizon         & $T$      & $1.0$ \\
\hline
\end{tabular}
\caption{Market parameters, American put under Black--Scholes.}
\label{tab:bs-params}
\end{table}

The Cox--Ross--Rubinstein binomial price with $N=2000$ time steps yields
$P_{\mathrm{ref}} = 0.11973$ (the intrinsic value at $t=0$ is
$K-S_0 = 0.1$). The solver returns $-Y_0 = 0.11638$, an absolute gap of
$-0.0034$ and a relative deviation of $-2.8\%$.

Figure~\ref{fig:bs-diagnostics} collects the three diagnostics. Panel
\ref{fig:bs-traj} displays, on the left, three realised value paths $Y_t$ with
the corresponding moving obstacle $-\xi_t$; the value process stays below the
obstacle pathwise and makes contact during sufficiently in-the-money
excursions. On the right, the same panel shows the realised control process
$Z_t$, which fluctuates around $-\sigma \cdot \partial_x u$ and is mostly
negative (the put has nonpositive delta in the state $X_t = \log S_t$). Panel
\ref{fig:bs-tau} is the histogram of the realised first-contact times
$\tau^\ast$ of Theorem~\ref{thm:rbsde-representation} across $1024$
simulated paths. Roughly $41\%$ of paths exit strictly before $T$, with mean
conditional stopping time $\bar\tau^\ast \approx 0.47$.
The cluster of values just below $T$ visible in the histogram is largely
a consequence of the imposed terminal contact $Y_T=-\xi_T$, since the
contact time detector therefore identifies paths near the terminal step. Panel \ref{fig:bs-loss} reports the training loss for
the terminal step $n=N-2$ on a log scale; the loss drops by more than two
orders of magnitude and plateaus near $10^{-4}$. The losses at the other time
steps are of the same shape and reach smaller magnitudes (closer to
$10^{-6}$), since the conditional-expectation regression smooths the obstacle
non-smoothness as it propagates backward.

\begin{figure}[!htb]
\centering
\begin{subfigure}{0.95\linewidth}
  \centering
  \safeincludegraphics[width=\linewidth]{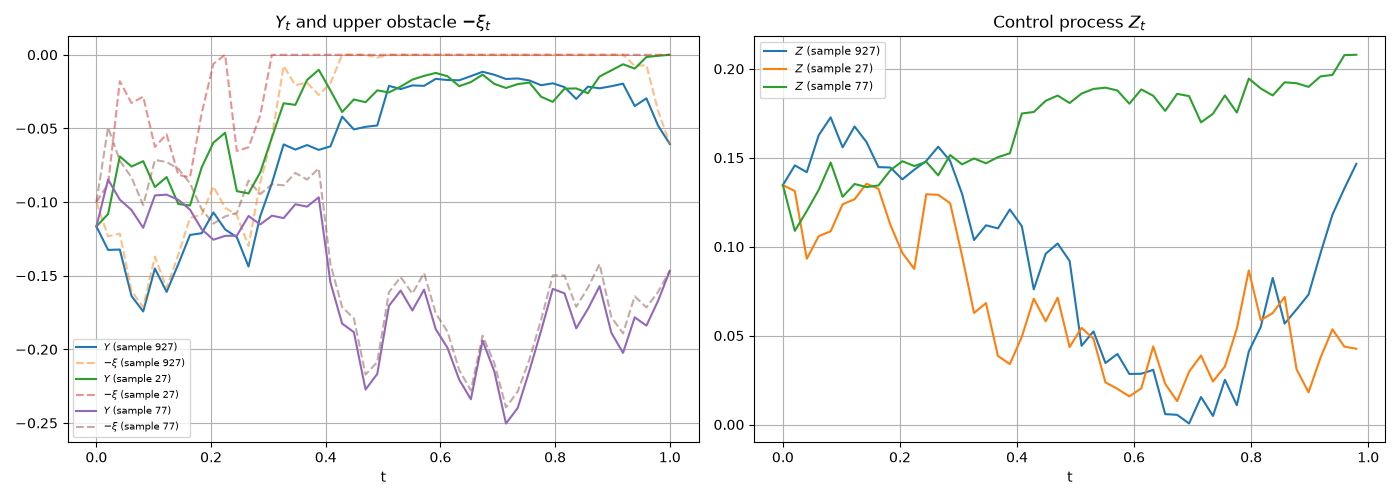}
  \caption{Three sample paths: left, $Y_t$ and obstacle $-\xi_t$; right,
  control process $Z_t$.}
  \label{fig:bs-traj}
\end{subfigure}\\[1ex]
\begin{subfigure}{0.45\linewidth}
  \centering
  \safeincludegraphics[width=\linewidth]{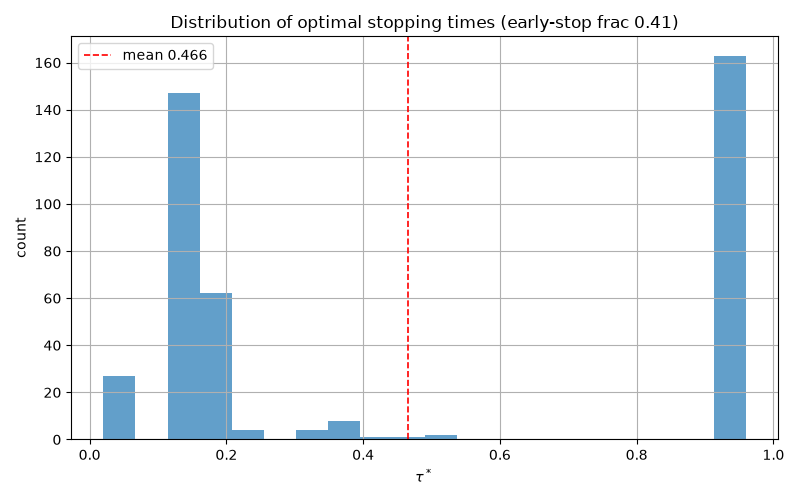}
  \caption{Histogram of $\tau^\ast$ over $1024$ paths
  (early-stop fraction $\approx 0.41$).}
  \label{fig:bs-tau}
\end{subfigure}\hfill
\begin{subfigure}{0.45\linewidth}
  \centering
  \safeincludegraphics[width=\linewidth]{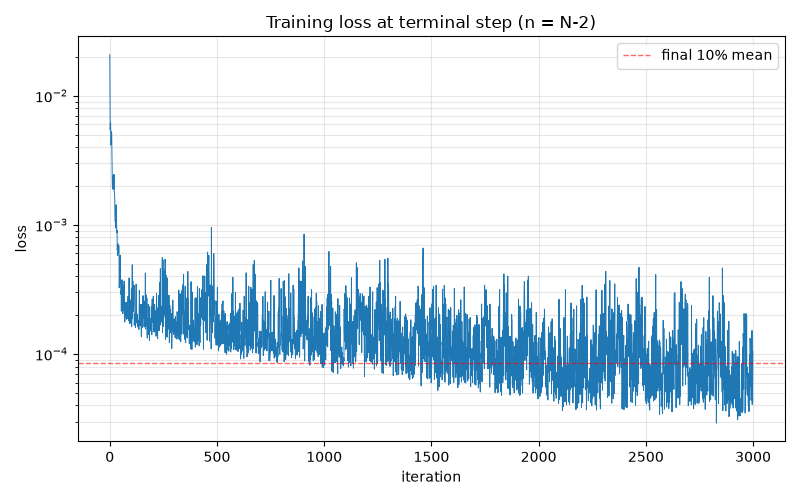}
  \caption{Training loss at the terminal step $n = N-2$ (log scale).}
  \label{fig:bs-loss}
\end{subfigure}
\caption{Black--Scholes American put: diagnostics from a single trained run.}
\label{fig:bs-diagnostics}
\end{figure}

\subsection{Bounded discount rate ambiguity}
\label{subsec:num-5-1}

This case applies the scheme to the driver of Section~\ref{subsec:bounded-beta},
$g(y) = \sup_{\beta \in [\underline\beta,\overline\beta]} \{-\beta y\}$, which
is convex in $y$ but no longer affine. Both branches of the sup become active
only when $Y$ visits both signs, so a nonnegative put payoff would pin
$Y\le0$ and collapse the band to $\overline\beta$. We therefore take the
collared obstacle
\[
\xi_t=c\tanh\bigl(X_t/c\bigr),
\qquad dX_t=dW_t,
\]
which is sign changing, and bounded by $c$ as required by
\ref{assump:obstacle}. It is smooth with bounded derivatives of every order, so
it satisfies \ref{assump:R1} and \ref{assump:R2} as well, and this is the one
test case for which the rate in Proposition~\ref{thm:main_convergence} is
available. Parameters are listed in Table~\ref{tab:5_1-params}.

\begin{table}[!htb]
\centering
\begin{tabular}{l l l}
\hline
Forward                  & $dX_t = dW_t$              & $X_0 = 0$ \\
Obstacle process         & $\xi_t = c\tanh(X_t/c)$    & sign-changing, $|\xi|\le c$ \\
Collar level             & $c$                        & $1.0$ \\
Lower discount bound     & $\underline\beta$          & $0.0$ \\
Upper discount bound     & $\overline\beta$           & $0.1$ \\
Time horizon             & $T$                        & $1.0$ \\
\hline
\end{tabular}
\caption{Forward, obstacle and discount band, Section~\ref{subsec:bounded-beta} case.}
\label{tab:5_1-params}
\end{table}

The solver returns $Y_0 = -0.0846$. Across interior time steps, $44.7\%$ of
$(t,\omega)$ samples have $Y_t \ge 0$, confirming that both branches of the sup
are exercised. Evaluated on $8192$ simulated paths, a fraction $0.54$ of paths
make first contact strictly before $T$, with mean $\bar\tau^\ast \approx 0.43$.

Figure~\ref{fig:51-diagnostics} reports the diagnostics. Panel
\ref{fig:51-traj} shows, on the left, the value process $Y_t$ together with
the moving obstacle $-\xi_t$ for three sample paths; the value stays below the
obstacle and makes contact on excursions where $X_t$ becomes sufficiently
negative. The collar flattens the obstacle for large $|X_t|$, so contact occurs
earlier than it would for an unbounded payoff. On the right, the same panel
shows the control process $Z_t$, which carries the trend of the sign-changing
obstacle and is damped where the collar saturates. Panel
\ref{fig:51-tau} is the histogram of $\tau^\ast$ over the paths making contact
before $T$, the terminal step being excluded by construction. The mass is
concentrated in the interior of $[0,T]$, rising to a peak around $t=0.5$ and
decaying thereafter, with mean $0.43$ and median $0.52$. Panel
\ref{fig:51-loss} reports the terminal-step training loss; its order of
magnitude and shape match Figure~\ref{fig:bs-loss}, and the losses at earlier
steps are smaller still. The band sup-structure does not destabilise the
recursion.

\begin{figure}[!htb]
\centering
\begin{subfigure}{0.95\linewidth}
  \centering
  \safeincludegraphics[width=\linewidth]{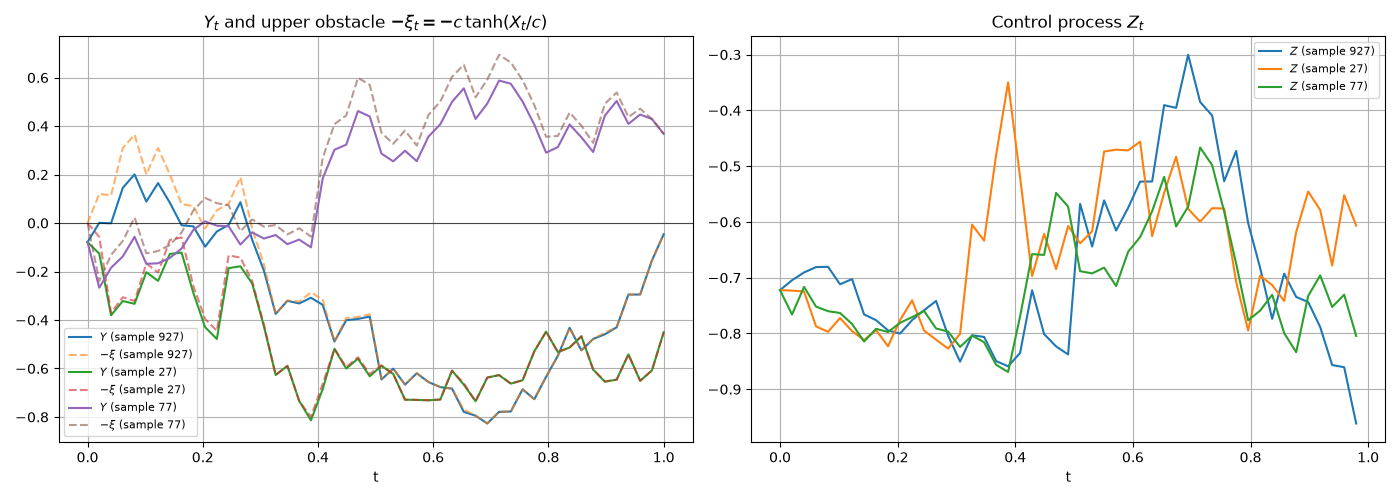}
  \caption{Three sample paths: left, $Y_t$ and obstacle
  $-\xi_t = -c\tanh(X_t/c)$; right, control process $Z_t$.}
  \label{fig:51-traj}
\end{subfigure}\\[1ex]
\begin{subfigure}{0.45\linewidth}
  \centering
  \safeincludegraphics[width=\linewidth]{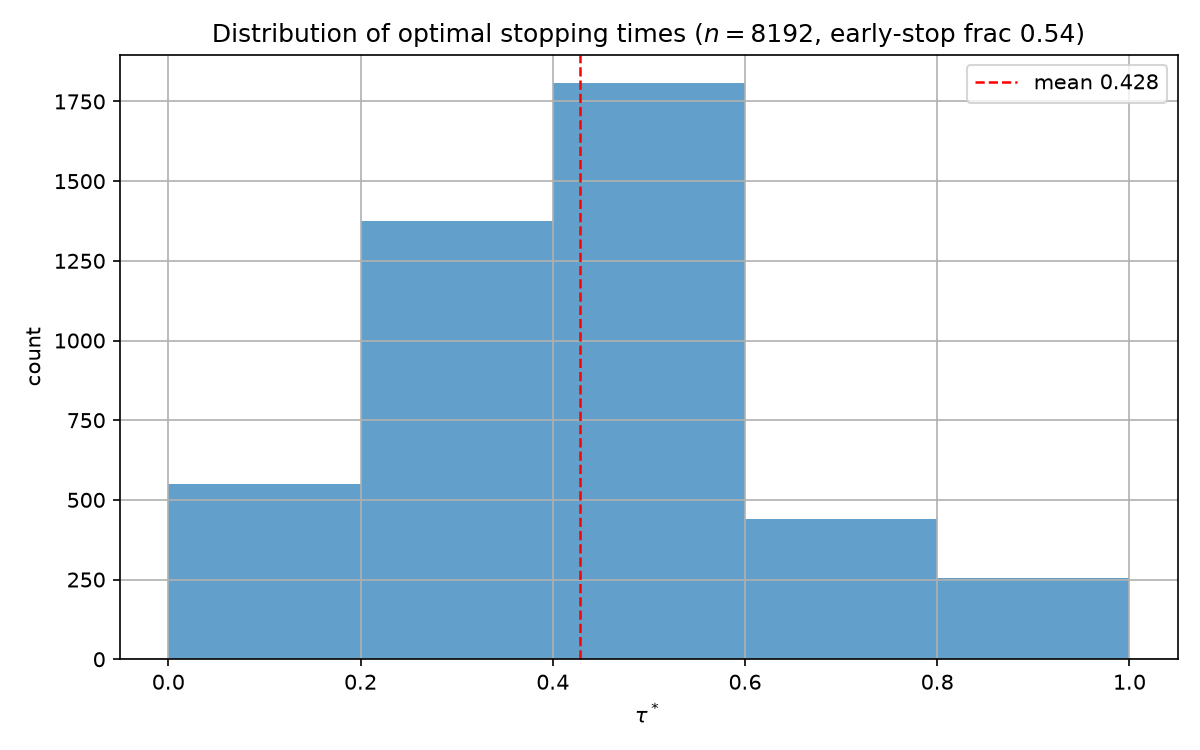}
  \caption{Histogram of $\tau^\ast$ over $8192$ paths.}
  \label{fig:51-tau}
\end{subfigure}\hfill
\begin{subfigure}{0.45\linewidth}
  \centering
  \safeincludegraphics[width=\linewidth]{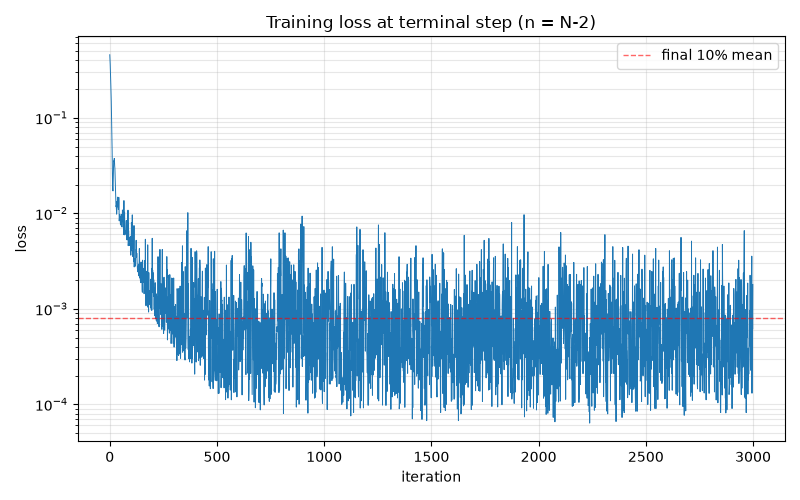}
  \caption{Training loss at the terminal step (log scale).}
  \label{fig:51-loss}
\end{subfigure}
\caption{Bounded discount rate ambiguity (Section~\ref{subsec:bounded-beta}):
diagnostics from a single trained run.}
\label{fig:51-diagnostics}
\end{figure}

\subsection{Geometric American put under entropic-discount ambiguity}
\label{subsec:num-5-4}

This case is the non-Lipschitz target of the convergence analysis in
Section~\ref{sec:numerical_convergence}: the driver of Section~\ref{subsec:put}
carries the quadratic in $z$ term $(\overline\gamma/2)|z|^2$ that lies outside
the Lipschitz framework and motivates the discrete reflection and truncation analysis of Section~\ref{sec:numerical_convergence}. We
demonstrate that the scheme, combined with truncation and the upper reflection
projection, handles the superlinear driver and that the entropic ambiguity has
a measurable effect on the conservative put price relative to
Section~\ref{subsec:num-bs}.

In the liability region $Y \le 0$ the driver of Section~\ref{subsec:put} reduces to
$g(y,z) = -\overline\delta y + (\overline\gamma/2)|z|^2$. We take
$\overline\delta = r$, so the worst case discount equals the risk-free rate,
and $\overline\gamma = 5$ (motivated below). The forward dynamics, payoff and
market parameters are inherited from Table~\ref{tab:bs-params}; the additional
ambiguity parameters are in Table~\ref{tab:5_4-params}.

\begin{table}[!htb]
\centering
\begin{tabular}{l l l}
\hline
Worst-case discount  & $\overline\delta$ & $0.05$ \\
Entropic radius      & $\overline\gamma$ & $5.0$  \\
\hline
\end{tabular}
\caption{Ambiguity parameters, Section~\ref{subsec:put} case.}
\label{tab:5_4-params}
\end{table}

A sweep over $\overline\gamma \in \{0.5, 1, 2, 5, 10\}$ at fixed market
parameters yields $-Y_0$ equal to $0.115$, $0.115$, $0.112$, $0.109$ and
$0.107$ respectively. Below $\overline\gamma \approx 2$ the entropic effect is
buried in the optimal stopping bias; above $\overline\gamma \approx 10$ the
conservative price approaches the intrinsic value $K - S_0 = 0.1$ and the
continuation region collapses. The value $\overline\gamma = 5$ keeps the price
clearly above intrinsic while the quadratic term contributes meaningfully. At
$\overline\gamma = 5$ the solver returns $-Y_0 = 0.10881$, a $-6.5\%$ shift
from the Black--Scholes baseline of Section~\ref{subsec:num-bs}.

Figure~\ref{fig:54-diagnostics} reports the diagnostics. Panel
\ref{fig:54-traj} shows, on the left, the value process $Y_t$ with the moving
obstacle $-\xi_t = -(K-S_t)^+$ for three sample paths; the obstacle constraint
is respected without measurable violation. On the right, the same panel shows
the realised control process $Z_t$, which is the quantity penalised
quadratically by the entropic driver. Panel \ref{fig:54-tau} is the histogram
of $\tau^\ast$. Comparing with Figure~\ref{fig:bs-tau}, two changes are
visible. First, the cluster just below $T$ remains as a consequence of the imposed
terminal contact $Y_T=-\xi_T$. Second, the rest of the distribution shifts toward earlier times:
the early stop fraction rises from $0.41$ in Section~\ref{subsec:num-bs} to $0.82$ here.
This empirical shift is consistent with the increased contribution of the
quadratic $z$ term in the continuation region, although the experiment does
not by itself establish a comparative statics result for the optimal stopping
boundary.
Panel \ref{fig:54-loss} reports the terminal-step training loss; it has the
same shape and the same order of magnitude as in
Figures~\ref{fig:bs-loss}--\ref{fig:51-loss}, with losses at the earlier time
steps smaller by an order of magnitude. The superlinear driver therefore does
not destabilise training under the hyperparameters of
Table~\ref{tab:nn-params}.

\begin{figure}[!htb]
\centering
\begin{subfigure}{0.95\linewidth}
  \centering
  \safeincludegraphics[width=\linewidth]{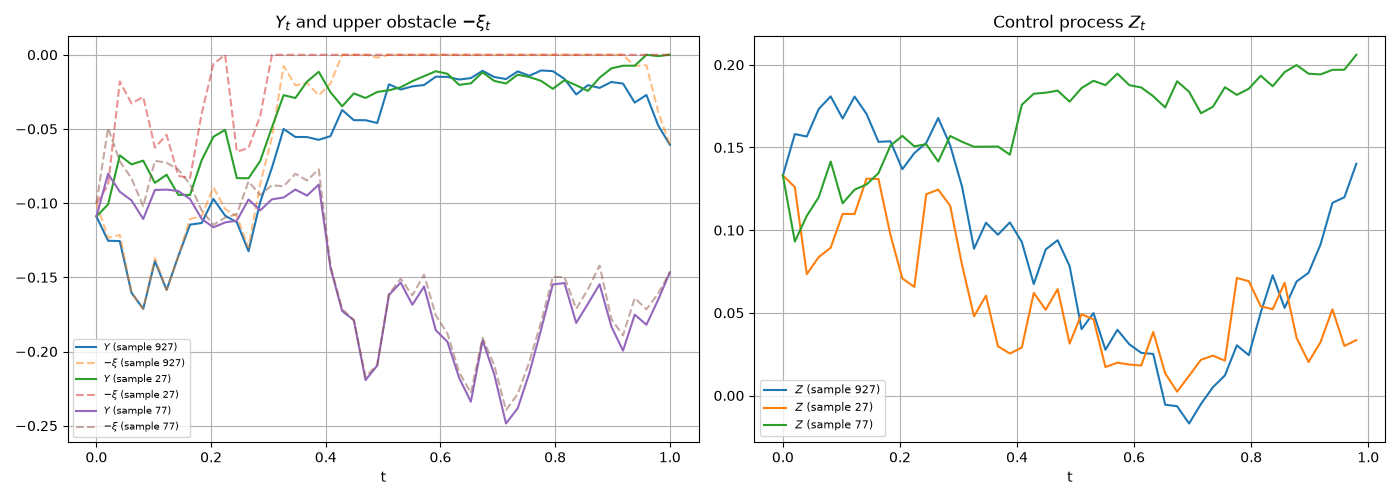}
  \caption{Three sample paths: left, $Y_t$ and obstacle $-\xi_t = -(K-S_t)^+$;
  right, control process $Z_t$.}
  \label{fig:54-traj}
\end{subfigure}\\[1ex]
\begin{subfigure}{0.45\linewidth}
  \centering
  \safeincludegraphics[width=\linewidth]{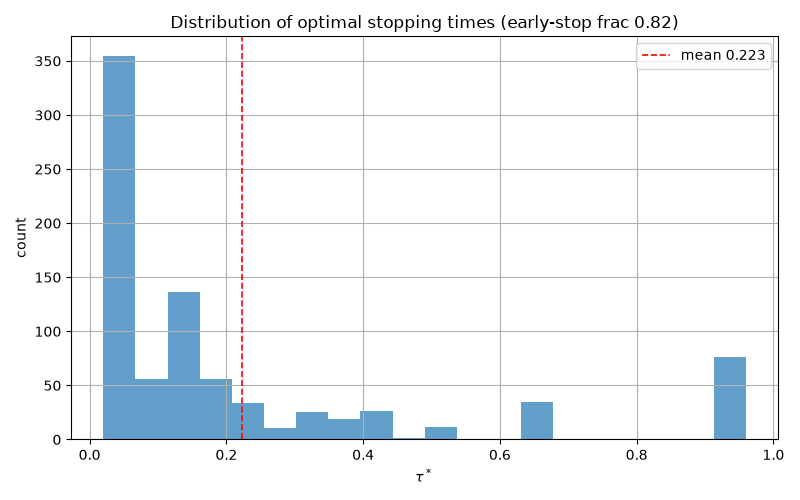}
  \caption{Histogram of $\tau^\ast$ over $1024$ paths
  (early-stop fraction $\approx 0.82$).}
  \label{fig:54-tau}
\end{subfigure}\hfill
\begin{subfigure}{0.45\linewidth}
  \centering
  \safeincludegraphics[width=\linewidth]{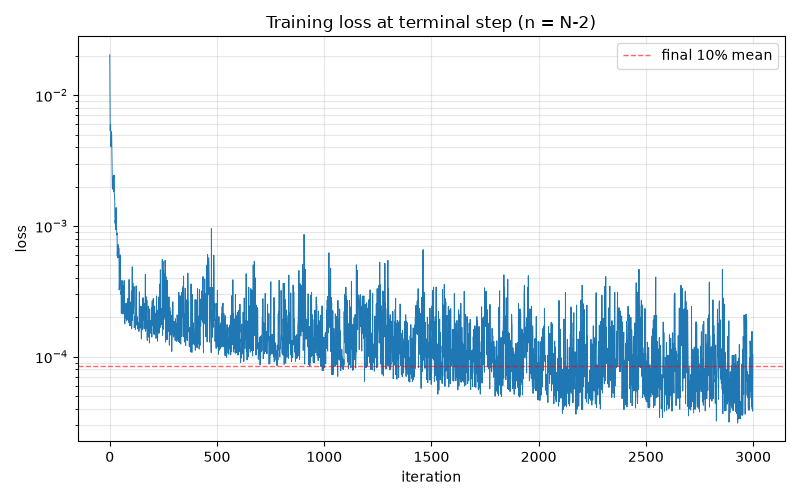}
  \caption{Training loss at the terminal step (log scale).}
  \label{fig:54-loss}
\end{subfigure}
\caption{Entropic-discount American put (Section~\ref{subsec:put}): diagnostics
from a single trained run at $\overline\gamma = 5$.}
\label{fig:54-diagnostics}
\end{figure}

\subsection{Numerical verification of Theorem~\ref{prop: properties}}
\label{sec:numerics:properties}

For each of the cases of Section~\ref{subsec:num-bs}--\ref{subsec:num-5-4} we test
on the trained solver the three properties of Theorem~\ref{prop: properties}
that can be read off from the value $\varphi_0(\xi) = Y_0$:
\begin{itemize}
    \item[(a)] $\varphi_0(\xi + m) \le \varphi_0(\xi)$ for a constant shift
    $m = 0.2$;
    \item[(b)] $\varphi_0(\xi + m) \ge \varphi_0(\xi) - m$ for the same shift;
    \item[(e)] $\varphi_0\bigl(\tfrac12\xi^1 + \tfrac12\xi^2\bigr)
    \ge \tfrac12\,\varphi_0(\xi^1) + \tfrac12\,\varphi_0(\xi^2)$ for a basket
    of two obstacles.
\end{itemize}
Under the standing assumptions of Section~\ref{sec: assumpt}, property (a) holds for every
admissible driver, while property (b) holds for the drivers considered here
because they are decreasing in $y$. Both properties must therefore be
reproduced by the solver in all three cases.
Concavity (e), on the other hand, is established only when the driver is
affine in $(y,z)$. This is the situation of Section~\ref{subsec:num-bs}, where
$g(y) = -r y$; the inequality (e) is therefore expected. In
Section~\ref{subsec:num-5-1} and Section~\ref{subsec:num-5-4} the driver is convex but not
affine, so Theorem~\ref{prop: properties} makes no prediction about the
sign of $\varphi_0(\tfrac12\xi^1 + \tfrac12\xi^2)
- \tfrac12[\varphi_0(\xi^1) + \varphi_0(\xi^2)]$, and the corresponding entry
of the table is reported only as an empirical observation on the chosen
basket.

For each case we train four networks, on the base obstacle and on its shifted,
alternate, and basketed versions, at the hyperparameters of
Table~\ref{tab:nn-params}. We say that an inequality is satisfied numerically
when the corresponding signed margin (defined in the caption of
Table~\ref{tab:properties}) is at least $-0.02$; this threshold is twice the
variation observed across independent training replicates. The results are
collected in Table~\ref{tab:properties}.

The two unconditional inequalities (a) and (b) are satisfied in every case.
Concavity (e) holds in Section~\ref{subsec:num-bs}, in agreement with
Theorem~\ref{prop: properties}(e). In Sections~\ref{subsec:num-5-1}
and~\ref{subsec:num-5-4} the signed margin is within the numerical tolerance, so
the inequality is neither clearly satisfied nor clearly violated on the chosen
baskets. This is consistent with Theorem~\ref{prop: properties}, which asserts
concavity only for affine drivers and makes no prediction here; the experiments
neither confirm nor refute it for the two non-affine cases.

\begin{table}[!htb]
\centering
\begin{tabular}{l c c c}
\hline
Case & (a) & (b) & (e) \\
\hline
Section~\ref{subsec:num-bs}              & $+0.208$ & $+0.012$ & $+0.0044$ \\
Section~\ref{subsec:num-5-1}             & $+0.212$ & $+0.008$ & $-0.009$  \\
Section~\ref{subsec:num-5-4}             & $+0.214$ & $+0.006$ & $+0.011$  \\
\hline
\end{tabular}
\caption{Signed margins of the three inequalities of
Theorem~\ref{prop: properties} at $m=0.2$ and $\lambda=\tfrac12$. The
quantities tabulated are
$\varphi_0(\xi) - \varphi_0(\xi + m)$ for (a),
$\varphi_0(\xi + m) - \varphi_0(\xi) + m$ for (b), and
$\varphi_0(\tfrac12\xi^1 + \tfrac12\xi^2)
 - \tfrac12[\varphi_0(\xi^1) + \varphi_0(\xi^2)]$ for (e). The numerical
tolerance is $0.02$. Baskets: Section~\ref{subsec:num-bs} and
Section~\ref{subsec:num-5-4} use puts with strikes $K_1 = 1.1$, $K_2 = 0.9$;
Section~\ref{subsec:num-5-1} uses $\xi^1_t = c\tanh(X_t/c)$ and
$\xi^2_t = c\tanh\bigl((X_t^2 - 0.5)/c\bigr)$, the collar of
Section~\ref{subsec:num-5-1} applied to both.}
\label{tab:properties}
\end{table}

\bibliographystyle{plain}
\bibliography{references}

\end{document}